\documentclass[12pt]{article}
\usepackage{times}
\usepackage[a4paper,hmargin=1in,vmargin=1.1in]{geometry}
\usepackage{amsthm}

\usepackage{multirow}
\usepackage{enumerate}
\usepackage{comment}
\usepackage{marginnote}
\usepackage{url}
\usepackage{subfig}
\usepackage{amsfonts}
\usepackage{amsmath}
\usepackage{amssymb}
\usepackage{wrapfig}

\usepackage{graphicx}
\graphicspath{{./Figs/}}
\usepackage{paralist}
\usepackage[active]{srcltx}
\usepackage{color}
\usepackage{algorithm}
\usepackage{algpseudocode}
\usepackage{setspace}

\algrenewcommand\algorithmicrequire{\textbf{\quad Input:}}
\algrenewcommand\algorithmicensure{\textbf{\quad Output:}}
\newenvironment{proof sketch}[1]{\noindent {\emph{Proof sketch of #1:}}}{\hfill \qed}

\newtheorem{theorem}{Theorem}
\newtheorem{proposition}[theorem]{Proposition}
\newtheorem{lemma}[theorem]{Lemma}

\newtheorem{definition}{Definition}

\newtheorem{coro}[theorem]{Corollary}
\newtheorem{rem}{Remark}

\newcommand{\eqdf}{\triangleq}
\newcommand{\eps}{\varepsilon}

\newcommand{\poly}{{\rm poly}}

\newcommand{\sol}{\textit{sol}}

\newcommand{\alg}{\textsc{alg}}

\newcommand{\degree}{\text{\textit{deg}}}
\newcommand{\dist}{\text{\textit{dist}}}
\newcommand{\gain}{\text{\textit{gain}}}

\newcommand{\dlocal}{\text{\textsc{DistLocal}}}
\newcommand{\clocal}{\text{\textsc{CentLocal}}}

\newcommand{\orad}{\text{\textsc{o-rad}}}
\newcommand{\obr}{\text{\textsc{obr}}}
\newcommand{\rad}{\text{\textit{rad}}}
\newcommand{\reach}{\text{\textit{reach}}}
\newcommand{\Reach}{\text{\textit{Reach}}}

\newcommand{\mis}{\text{\textsc{mis}}}

\newcommand{\igmis}{\text{\textsc{ig-mis}}}

\newcommand{\lmis}{\text{\textsc{l-mis}}}
\newcommand{\ao}{\text{\textsc{ao}}}

\newcommand{\mm}{\text{\textsc{mm}}}
\newcommand{\mcm}{\text{\textsc{mcm}}}

\newcommand{\mum}{\text{\textsc{mcm}}}
\newcommand{\mwm}{\text{\textsc{mwm}}}

\newcommand{\wmin}{w_{\min}}

\newcommand{\pcolor}{\text{$\Delta^2$-\textsc{Color}}}
\newcommand{\dcolor}{\text{$(\Delta+1)$\textsc{-Color}}}
\newcommand{\ccolor}{\text{$c$\textsc{-Color}}}

\newcommand{\probe}{\mathsf{probe}}

\newcommand{\NN}{{\mathbb{N}}}

\newcommand{\oracle}{{\mathcal O}}
\newcommand{\proca}{{\mathcal A}}
\newcommand{\aug}{{\textsc{aug}}}

\def\mnFONT{\footnotesize}

\newcommand{\red}{\color{black}}

\begin{document}
\title{
Best of Two Local Models:\\
Centralized local and Distributed local Algorithms%
}

\author{ %
Guy Even\thanks{School of Electrical Engineering, Tel-Aviv
Univ., Tel-Aviv 69978, Israel.
\protect\url{{guy,medinamo,danar}@eng.tau.ac.il}.}
\and Moti Medina$^*$\thanks{M.M was partially funded
by the Israeli Ministry of Science and Technology.}
\and Dana Ron$^*$\thanks{Research supported by the Israel Science Foundation grant number 671/13.}
}
\date{}

\maketitle
\begin{abstract}
  We consider two models of computation: centralized
  local algorithms and local distributed algorithms.
  Algorithms in one model are adapted to the other model
  to obtain improved algorithms.

  Distributed vertex coloring is employed to design
  improved centralized local algorithms for: maximal
  independent set, maximal matching, and an approximation
  scheme for maximum (weighted) matching over bounded
  degree graphs. The improvement is threefold: the
  algorithms are deterministic, stateless, and the number
  of probes grows polynomially in $\log^* n$, where $n$ is the number
  of vertices of the input graph.

  The recursive centralized local improvement technique
  by Nguyen and Onak~\cite{onak2008} is employed to
  obtain an improved distributed approximation scheme for
  maximum (weighted) matching.
The improvement is twofold: we reduce the number of rounds from $O(\log n)$ to
$O(\log^*n)$ for a wide range of instances
and, our algorithms are
deterministic rather than randomized.
\end{abstract}


\section{Introduction}

{\em Local Computation Algorithms\/}, as defined by Rubinfeld et
al.~\cite{shaiics2011}, are algorithms that answer queries regarding (global)
solutions to computational problems by performing local (sublinear time) computations
on the input. The answers to all queries must be consistent with a single solution regardless of the number of possible solutions.  To make this notion
concrete, consider the {\em Maximal Independent Set\/} problem, which we denote by
\mis.  Given a graph $G = (V,E)$, the local algorithm \alg\ gives the
illusion that it ``holds'' a specific maximal independent set $I \subseteq V$.  Namely, given
any vertex $v$ as a query, \alg\ answers whether $v$ belongs to $I$ even though
\alg\ cannot read all of $G$, cannot store the entire solution $I$, and cannot even remember all the answers to previous queries.  In order to answer such queries, $\alg$
can probe the graph $G$ by asking about the neighbors of a vertex of its choice.

\sloppy
A local {\red computation} algorithm may be randomized, so that the solution
according to which it answers queries may depend on its internal coin flips. However,
the solution should not depend on the sequence of the queries {\red(this property is
  called query order obliviousness~\cite{shaiics2011}).}  We measure the performance
of a local computation algorithm by the following criteria: the maximum number of
probes it makes to the input per query, the success probability over any sequence of
queries, and the maximum space it uses between queries%
\footnote{In the RAM model, the running time per query of our algorithms is at most
  $\poly(\textit{ppq})\cdot \log \log n$, where $\textit{ppq}$ is the maximum number
  of probes per query and $n = |V|$.  }
.
It is desired that both the probe complexity
and the space complexity of the algorithm be sublinear in the size of the graph
(e.g., ${\rm polylog}(|V|)$), and that the success probability be $1-1/\poly(|V|)$.
It is usually assumed that the maximum degree of the graph is upper-bounded by a
constant, but our results are useful also for non-constant upper bounds (see also~\cite{DBLP:journals/corr/ReingoldV14}).
For a formal definition of local algorithms in the context of graph problems,
which is the focus of this work, see Subsection~\ref{sec:cmodel}.

The motivation for designing local computation algorithms is that local computation algorithms capture difficulties with very large inputs.
A few examples include:
\begin{inparaenum}[(1)]
  \item Reading the entire input is too costly if the input is very large.
  \item In certain situations one is interested in a very small part of a complete solution.
  \item Consider a setting in which different uncoordinated servers need to answer queries about a very large input stored in the cloud.
      The servers do not communicate with each other, do not store answers to previous queries, and want to minimize their accesses to the input. Furthermore, the servers answer the queries consistently.
\end{inparaenum}

Local computation algorithms have been designed for various graph (and hypergraph) problems,
including the abovementioned \mis~\cite{shaiics2011,shaisoda2012},
hypergraph coloring~\cite{shaiics2011,shaisoda2012},
maximal matching~\cite{shaiicalp2012}
and (approximate) maximum matching~\cite{shaiapprox2013}.
Local computation algorithms also appear  implicitly in works on sublinear
approximation algorithms for various graph parameters, such as the size of
a minimum vertex cover~\cite{parnasron,onak2008,yoshida2009improved,dana2012}.
Some of these implicit results are very efficient in terms of their probe complexity
(in particular, it depends on the maximum degree and not on $|V|$)
but do not give the desired $1-1/\poly(|V|)$ success probability. We compare our
results to both the explicit and implicit relevant known results.

As can be gleaned from the definition in~\cite{shaiics2011},
local computation algorithms are closely related to {\em Local Distributed Algorithms\/}.
We discuss the similarities and differences in more detail in Subsection~\ref{subsec:rel-clocal-dlocal}.
In this work, we exploit this relation in two ways. First, we
use techniques from the study of local distributed algorithms to obtain better local computation algorithms.
Second, we apply techniques
from the study of local computation algorithms (more precisely, local computation algorithms that
are implicit within sublinear approximation algorithms) to
obtain a new result in distributed computing.

In what follows we denote the aforementioned local computation model by \clocal\
(where the ``\textsc{Cent}'' stands for ``centralized'')
and the distributed (local) model
 by \dlocal\ (for a formal definition of the latter,
see Subsection~\ref{sec:dlocal-def}).
We denote the number of vertices in the input graph by $n$ and the maximum degree by
$\Delta$.

\subsection{On the relation between \clocal\ and \dlocal}\label{subsec:rel-clocal-dlocal}
The \clocal\ model is centralized in the sense that there is a single central
algorithm that is provided access to the whole graph. This is as opposed to the
\dlocal\ model in which each processor resides in a graph vertex $v$ and can obtain
information only about the neighborhood of $v$. Another important difference is in
the main complexity measure. In the \clocal\ model, one counts the number of probes
that the algorithm performs per query, while in the \dlocal\ model, the number of
rounds of communication is counted. This implies that a \dlocal\ algorithm 
obtains information about a ball centered at a vertex, where the radius of the ball
is the number of rounds of communication. On the other hand, in the case of a
\clocal\ algorithm, it might choose to obtain information about different types of
neighborhoods so as to save in the number of probes. Indeed (similarly to what was
observed in the context of sublinear approximation algorithms~\cite{parnasron}),
given a \dlocal\ algorithm for a particular problem with round complexity $r$, we
directly obtain a \clocal\ algorithm whose probe complexity is $O(\Delta^r)$ where
$\Delta$ is the maximum degree in the graph. However, we might be able to obtain
lower probe complexity if we do not apply such a black-box reduction.  In the other
direction, \clocal\ algorithms with certain properties, can be transformed into
\dlocal\ algorithms (e.g., a deterministic \clocal\-algorithm in which probes are
confined to an $r$-neighborhood of the query).

\subsection{The Ranking Technique}
The starting point for our results in the \clocal\ model is the
\emph{ranking} technique~\cite{onak2008,yoshida2009improved,shaisoda2012,shaiicalp2012,shaiapprox2013}. To exemplify this, consider, once again, the \mis\ problem.
A very simple (global ``greedy'') algorithm for this problem works by
selecting an arbitrary ranking of the vertices and initializing $I$ to be empty.
The algorithm then considers the vertices one after the other according to their
ranks and adds a vertex to $I$ if and only if it does not neighbor any vertex
already in $I$. Such an algorithm can be ``localized'' as follows. For a fixed
ranking of the vertices (say, according to their IDs), given a query on a vertex $v$,
the local algorithm performs a {\em restricted\/} DFS starting from $v$. The
restriction is that the search continues only on paths with monotonically decreasing
ranks. The local algorithm then simulates the global one on the subgraph induced by
this restricted DFS.

The main problem with the above local algorithm is that the number of probes it performs
when running the DFS may be very large. Indeed, for some rankings (and queried vertices), the number
of probes is linear in $n$. In order to circumvent this problem,
{\em random\/} rankings were studied~\cite{onak2008}. This brings up two questions, which were studied in previous works,
both for the \mis\ algorithm described above and for other ranking-based
algorithms~\cite{onak2008,yoshida2009improved,shaisoda2012,shaiicalp2012,shaiapprox2013}.
The first is to bound
the number of probes needed to answer a query with high probability. The second is
how to efficiently store a random ranking between queries.

\subsection{Our Contributions}
\label{subsec:results-clocal}
In this section we overview the techniques we use and the results we obtained based
on these techniques. See the tables in Section~\ref{subsec:rel-work} for a precise
statement of the results.

\paragraph{Orientations with bounded reachability.}
Our first conceptual contribution is a simple but very useful observation. Rather
than considering vertex rankings, we suggest to consider {\em acyclic orientations\/}
of the edges in the graph. Such orientations induce partial orders over the vertices,
and partial orders suffice for our purposes. The probe complexity induced by a given
orientation translates into a combinatorial measure, which we refer to as the {\em
  reachability\/} of the orientation. Reachability of an acyclic orientation is the
maximum number of vertices that can be reached from any start vertex by
directed paths (induced by the orientation). This leads us to the quest for a
\clocal\ algorithm that computes an orientation with bounded reachability.

\paragraph{Orientations and colorings.}
Our second conceptual contribution is that an orientation algorithm with bounded reachability can be based on a \clocal\ {\em coloring\/} algorithm.
Indeed, every vertex-coloring with $k$ colors induces an orientation with reachability
$O(\Delta^k)$. Towards this end, we design a \clocal\ coloring algorithm that applies techniques from \dlocal\ colorings algorithms~\cite{cole1986deterministic,goldberg1988parallel,linial,panconesi2010fast}.
 Our \clocal\  algorithm is deterministic, does not use any space between queries, performs
 $O(\Delta\cdot \log^* n+\Delta^3)$ probes per query, and computes a coloring
 with $O(\Delta^2)$ colors. (We refer to the problem of coloring a graph by $c$ colors as \ccolor.) Our coloring algorithm yields an orientation
 whose reachability is $\Delta^{O(\Delta^2)}$.
 For constant degree graphs, this implies $O(\log^* n)$ probes to obtain an orientation with constant reachability.
 As an application of this orientation algorithm, we also  design a \clocal\ algorithm for $(\Delta+1)$-coloring.

\paragraph{{\red Centralized local}  simulations of sequential algorithms.}
We apply a general transformation (similarly to what was shown in~\cite{shaisoda2012}) from global
algorithms with certain properties to local algorithms.
The transformation is based on our \clocal\ orientation with bounded reachability algorithm.
As a result we get
deterministic \clocal\ algorithms for \mis\ and maximal matching (\mm), which significantly
improve over previous work~\cite{shaiics2011,shaisoda2012,shaiicalp2012}, and the first \clocal\
algorithm for coloring with $(\Delta+1)$ colors. Compared to previous work,
for \mis\ and \mm\ the dependence on $n$ in the probe complexity is reduced from
${\rm polylog}(n)$ to $\log^*(n)$ and the space needed to store the state between queries is reduced from ${\rm
  polylog}(n)$ to zero.


\paragraph{Deterministic \clocal-algorithms for approximate maximum matching.}
We present $(1-\eps)$-approximation \clocal-algorithms for maximum cardinality
matching (\mcm) and maximum weighted matching (\mwm). 
Similarly to previous related work~\cite{onak2008,DBLP:conf/spaa/LotkerPP08,shaiapprox2013},
our algorithm for \mcm\ is based on the augmenting paths framework of Hopcroft and Karp~\cite{hopcroft1973n}.
Our starting point is a global/abstract algorithm
that works iteratively, where in each iteration it constructs a new matching (starting from the empty matching).
Each new matching is constructed based on a maximal set of vertex disjoint paths that are augmenting paths with respect to the previous matching. Such a maximal set is a maximal independent set (\mis) in the intersection graph over the augmenting paths.  
The question is how to simulate this global algorithm in a local/distributed fashion, and
in particular, how to compute the maximal independent sets  over the intersection graphs.

By using our \clocal\ \mis\ algorithm (over the intersection graphs), for the case of an
approximate \mcm, 
we reduce the dependence of the
probe-complexity on $n$ from ${\rm polylog}(n)$~\cite{shaiapprox2013} to $\poly(\log^*(n))$.
The space needed to store the state between queries is reduced from ${\rm
  polylog}(n)$ to $0$.
For the approximate \mwm\ algorithm we also build on
the parallel approximation algorithm of Hougardy and Vinkemeir~\cite{HV06}.

%


\paragraph{Deterministic \dlocal-algorithms for approximate maximum matching.}
We present $(1-\eps)$-approximation \dlocal-algorithms for \mcm\ and \mwm.  These
algorithms are based on a distributed simulation of the corresponding
\clocal-algorithms. For \mcm, we present a deterministic distributed
$(1-\eps)$-approximation algorithm. The number of rounds used by the algorithm is
$$\Delta^{O(1/\eps)} + O\left(\frac{1}{\eps^2}\right) \cdot\log^*(n).$$
For \mwm, we assume that edge weights are normalized as follows: the maximum edge
weight is $1$ and $\wmin$ denotes the minimum edge weight.
We present a deterministic distributed $(1-\eps)$-approximation algorithm. The number
of rounds used by the algorithm is
$$O\left(\frac{1}{\eps^2} \cdot \log
  \frac{1}{\eps}\right)\cdot \log^*n + \Delta^{O(1/\eps)}\cdot \log
\left(\min\{1/\wmin,n/\eps\} \right).$$

We briefly compare these results with previous results. The best previous algorithms
for both the unweighted and weighted cases are by Lotker, Patt-Shamir, and
Pettie~\cite{DBLP:conf/spaa/LotkerPP08}. For the unweighted case they give a
randomized $(1-\eps)$-approximation algorithm that runs in $O((\log(n))/\eps^{3})$
rounds with high probability%
\footnote{We say that an event occurs with high probability if it occurs with
  probability at least $1-\frac{1}{\poly(n)}$. }
 (w.h.p). Hence we get an improved result
when $\Delta^{O(1/\eps)} = o(\log(n))$. In particular, for
constant $\Delta$ and $\eps$, the number of rounds is
$O(\log^*(n))$.  Note that an $O(1)$-approximation of a
maximum matching in an $n$-node ring cannot be computed by
any
deterministic
distributed algorithm in $o(\log^*(n))$
rounds~\cite{czygrinow2008fast,lenzen2008leveraging}.  For
the weighted case, they give a randomized $(1/2-\epsilon)$-approximation algorithm
that runs in $O(\log(\eps^{-1}) \cdot \log(n))$ rounds (w.h.p).\footnote{Lotker,
  Patt-Shamir and Pettie remark~\cite[Sec.  4]{DBLP:conf/spaa/LotkerPP08} that a
  $(1-\eps)$-MWM can be obtained in $O(\eps^{-4}\log^2 n)$ rounds (using messages of
  linear size), by adapting the algorithm of Hougardy and Vinkemeir~\cite{HV06}.
}
Our \mwm\ approximation algorithm 
runs in significantly fewer rounds
for various settings of the parameters $\Delta$, $1/\eps$, and $1/\wmin$.
In particular, when they are constants, the number of rounds is
$O(\log^*(n))$.

\subsection{Detailed Comparison with Previous
Work}\label{subsec:rel-work}

\paragraph{Comparison to previous (explicit) \clocal\ algorithms.}
A comparison of our results with previous \clocal\ algorithms is summarized in
Table~\ref{tab:app}. The results assume that $\Delta$ and $\eps$ are constant. (The
dependence of the number of probes and space on $\Delta$ and $\eps$ is not explicit
in~\cite{shaiicalp2012,shaiapprox2013}. For \mis\ explicit dependencies appear
in~\cite{shaisoda2012}.  In recent work, Levi et al.~\cite{ReutRonittAnak14} show how
the exponential dependence on $\Delta$ can be reduced to quasi-polynomial in the case
of (exact) \mis\ and \mm.)  Explicit dependencies on $\Delta$ and $\eps$ in our
result appear in the formal statements within the paper.

\renewcommand{\arraystretch}{1.55}
\begin{table}[htb]
\scriptsize
\centering
\begin{tabular}{| c || c | c | c || c | c |}
\hline
\multirow{2}{*}{Problem} & \multicolumn{3}{|c ||}{Previous work } &  {Here (Deterministic, 0-Space)}\\
\cline{2-5}
 & Space & \# Probes & success prob. & \# Probes \\
\hline\hline
\mis & $O(\log^2 n)$ & $O(\log^3 n)$ & $1-\frac{1}{n}$~\cite{shaisoda2012} &  $O(\log^* n)$ [Coro.~\ref{coro:mis}]\\
\mm & $O(\log^3 n)$& $O(\log^3 n)$ & $1-\frac{1}{n}$ ~\cite{shaiicalp2012}&  $O(\log^* n)$ [Coro.~\ref{coro:mis}]\\
\pcolor & --- & --- & --- &  $O(\log^* n)$ [Thm.~\ref{thm:colorpoly}]\\
\dcolor & --- & --- & --- &  $O(\log^* n)$ [Coro.~\ref{coro:mis}]\\
$(1-\eps)$-\mcm & $O(\log^3 n)$ & $O(\log^4 n)$ &
$1-\frac{1}{n^2}$~\cite{shaiapprox2013} &  $(\log^* n)^{O(1)}$ [Thm.~\ref{thm:mm}]\\
$(1-\eps)$-\mwm & --- & --- & --- &  $\left(\Gamma\right)^{O(1)} \cdot
    (\log^*n)^{O(1)}$ [Thm.~\ref{thm:clocal mwm}]\\
\hline
\end{tabular}
\caption{\small A comparison between \clocal\ algorithms.
\mis\ denotes maximal independent set,  \mm\ denotes maximal matching, \mcm\ denotes maximum cardinality matching, and
  \mwm\ denotes maximum weighted matching.   
Our algorithms are deterministic and stateless (i.e., the space needed to store the state between queries is zero).
Since the dependence on $\Delta$ and $\eps$ is not explicit in~\cite{shaiicalp2012,shaiapprox2013},
all the results are presented under 
the assumption that $\Delta = O(1)$ and $\eps = O(1)$.
For weighted graphs, the ratio between the maximum to
minimum edge weight is denoted by $\Gamma$ (we may assume that $\Gamma\leq n/\eps$).
The $(1-\eps)$-\mwm\ \clocal-algorithm is of interest (even for $\Gamma=n/\eps$) because it serves as a
basis for the $(1-\eps)$-\mwm\ \dlocal-algorithm.
}
   \label{tab:app}
\end{table}
\renewcommand{\arraystretch}{1}

\paragraph{Comparison to previous \clocal\ oracles in sublinear approximation
  algorithms.}
A sublinear approximation algorithm for a certain graph
parameter (e.g., the size of a minimum vertex cover) is
given probe access to the input graph and is required to
output an approximation of the graph parameter with high
constant success probability. Many such algorithms work
by designing an {\em oracle\/} that answers queries (e.g.,
a query can ask: does a given vertex belong to a fixed
small vertex cover?). The sublinear approximation algorithm
estimates the graph parameter by performing (a small number
of) queries to the oracle. The oracles are essentially
\clocal\ algorithms but they tend to have constant error
probability.
Furthermore, the question of bounded space
needed to store the state between queries was not an issue
in the design of these oracles, since only few queries are
performed by the sublinear approximation algorithm. Hence,
they are not usually considered to be ``bona fide''
\clocal\ algorithms. A comparison of our results and these
oracles appears in Table~\ref{tab:app2}.

In the recent result of Levi et al.~\cite{ReutRonittAnak14} it is  shown how, based on the
sublinear approximation algorithms of Yoshida et al.~\cite{yoshida2009improved}
(referenced in the table), it is possible to
reduce the dependence on the failure probability, $\delta$,
from inverse polynomial to inverse poly-logarithmic.
In particular, they obtain a $(1-\eps)$-approximation 
\clocal\ algorithm for \mcm\ that performs $\Delta^{O(1/\eps^2)} \cdot \poly(\log n)$
probes, uses space of the same order, and succeeds with probability $1-1/\poly(n)$.
\renewcommand{\arraystretch}{1.55}
\begin{table}[htb]
\scriptsize
\centering
\begin{tabular}{| c || c | c | c|| c | c |c |}
\hline
\multirow{2}{*}{Problem} & \multicolumn{3}{|c ||}{Previous work } &  \multicolumn{2}{|c |}{Here}\\
\cline{2-6}
 & \# Probes & success prob. & apx. ratio & \# Probes & apx. ratio\\
\hline\hline
\mis &  $O(\Delta^4)\cdot\poly(\frac{1}{\delta},\frac{1}{\eps})$ & $1-\delta$& $1-\eps$ ~\cite{yoshida2009improved}& $\Delta^{O(\Delta^2)} \cdot \log^* n$ & 1 \\
\mm & $O(\Delta^4)\cdot\poly(\frac{1}{\delta},\frac{1}{\eps})$ & $1-\delta$ & $1-\eps$~\cite{yoshida2009improved}& $\Delta^{O(\Delta^2)} \cdot \log^* n$ & $1$\\
\mcm  & $\Delta^{O(1/\eps)}\cdot\poly(\frac{1}{\delta},\frac{1}{\eps})$ & $1-\delta$& $1-\eps$  ~\cite{yoshida2009improved}& $(\log^* n)^{O(1/\eps)}\cdot 2^{O(\Delta^{1/\eps})}$& $1-\eps$\\
\hline
\end{tabular}
\caption{\small A comparison between \clocal\ oracles in sub-linear approximation
  algorithms and our \clocal\ (deterministic) algorithms.
The former algorithms were designed to work with constant success probability and a bound
  was given on their expected probe complexity. When presenting them as \clocal\ algorithms we introduce a failure probability
  parameter, $\delta$, and bound their probe complexity in terms of $\delta$.
  Furthermore, the approximation ratios of the sublinear approximation algorithms were stated in additive terms, and we translate
  the results so as to get a multiplicative approximation.}
   \label{tab:app2}
\end{table}
\renewcommand{\arraystretch}{1}

\paragraph{Comparison to previous \dlocal\ algorithms for \mcm\ and \mwm.}
We compare our results to previous ones in Table~\ref{tbl}.
The first line refers to the aforementioned algorithm by Lotker, Patt-Shamir, and
Pettie~\cite{DBLP:conf/spaa/LotkerPP08} for the unweighted case.
The second line in Table~\ref{tbl}
refers to an algorithm of Nguyen and Onak~\cite{onak2008}.
As they observe, their algorithm for approximating the size
of a maximum matching in sublinear time can be transformed
into a randomized distributed algorithm that succeeds with constant
probability (say, $2/3$) and runs in $\Delta^{O(1/\eps)}$
rounds.
The third line refers to the aforementioned algorithm by Lotker, Patt-Shamir, and
Pettie~\cite{DBLP:conf/spaa/LotkerPP08} for the weighted case.
The fourth line refers to the algorithm by Panconesi and
Sozio~\cite{panconesi2010fast} for weighted matching. They devise a deterministic
distributed $(1/6-\eps)$-approximation algorithm
that runs in $O\left(
  \frac{\log^4 (n)}{\eps}\cdot \log (\Gamma)\right)$ rounds, where $\Gamma$ is the ratio between the
maximum to minimum edge weight.

We remark that the randomized \clocal-algorithm by Mansour
and Vardi~\cite{shaiapprox2013} for $(1-\eps)$-approximate
maximum cardinality matching in bounded-degree graphs can be transformed into a
randomized \dlocal-algorithm for $(1-\eps)$-approximate
maximum cardinality matching (whose success probability is $1-1/\poly(n)$).
Their focus is on bounding the number
of probes, which they show is polylogarithmic in $n$ for constant $\Delta$ and $\eps$.
To the best of our understanding, an analysis of the probe-radius of their
algorithm will not imply a \dlocal-algorithm that runs in fewer rounds
than the algorithm of
Lotker, Patt-Shamir, and Pettie~\cite{DBLP:conf/spaa/LotkerPP08}.

\renewcommand{\arraystretch}{2}
\begin{table*}[htb]
\scriptsize
\centering
\begin{tabular}{| c | c | c | c || c| c | c |c |}
\hline
\multicolumn{4}{|c ||}{Previous work } &  \multicolumn{2}{|c |}{Here (Deterministic)}\\
\cline{1-6}
problem & \# rounds & success prob. & apx. ratio. & \# rounds & apx. ratio.\\
\hline\hline
\multirow{2}{*}{\mcm} &$O(\frac{\log(n)}{\eps^{3}})$ & $1-\frac{1}{\poly(n)}$ & $1-\eps$ ~\cite{DBLP:conf/spaa/LotkerPP08}&  \multirow{2}{*}{$\Delta^{O\left(\frac{1}{\eps}\right)}+O\left(\frac{1}{\eps^2}\right)\cdot \log^*(n)$} & \multirow{2}{*}{$1-\eps$} \\
& $\Delta^{O(\frac{1}{\eps})}$ & $1-\Theta(1)$ & $1-\eps$ ~\cite{onak2008}& &   [Thm.~\ref{thm:distalg}]\\
%

\hline
\multirow{2}{*}{\mwm} &$O\left(\log(\eps^{-1}) \cdot \log(n)\right)$ & $1-\frac{1}{\poly(n)}$ & $1/2-\eps$ ~\cite{DBLP:conf/spaa/LotkerPP08}&  \multirow{2}{*}{$O\left(\frac{1}{\eps^2} \cdot \log \frac{1}{\eps}\right)\cdot \log^*n +
    \Delta^{O(1/\eps)}\cdot \log \left(\Gamma
         \right)$} & \multirow{2}{*}{$1-\eps$} \\
%
%
%
& $O\left( \frac{\log^4 (n)}{\eps}\cdot \log
(\Gamma)\right)$ & deterministic & $1/6-\eps$ ~\cite{panconesi2010fast}& &  [Thm.~\ref{thm:distalg mwm}]\\
%
\hline
\end{tabular}
\caption{\small A comparison between \mum\ and \mwm\
\dlocal\ algorithms. The ratio between the maximum to
minimum edge weight is denoted by $\Gamma$
(we may assume that $\Gamma\leq n/\eps$).}
   \label{tab:dist}\label{tbl}
\end{table*}
\renewcommand{\arraystretch}{1}

\section{Preliminaries}
\subsection{Notations}
Let $G=(V,E)$ denote an undirected graph, and $n(G)$
denote the number of vertices in $V$. We denote the degree of $v$ by $\degree_G(v)$.
Let $\Delta(G)$ denote the maximum degree, i.e., $\Delta(G)
\eqdf \max_{v\in V}\{\degree_G(v)\}$.
Let $\Gamma(v)$ denote the set of neighbors of $v\in V$.
The length of a path
equals the number of edges along the path. We denote the
length of a path $p$ by $|p|$. For $u,v \in V$ let
$\dist_G(u,v)$ denote the length of the shortest path
between $u$ and $v$ in the graph $G$. The ball of radius
$r$ centered at $v$ in the graph $G$ is defined by
\[
    B^G_r(v) \triangleq \{u \in V \mid \dist_G(v,u) \leq r\}\:.
\]
If the graph $G$ is clear from the context, we may drop it
from the notation, e.g., we simply write $n,m,\degree(v)$, or
$\Delta$.

For $k \in \NN^+$ and $n >0$, let $\log^{(k)} (n)$ denote
the $k$th iterated logarithm of $n$. Note that $\log ^{(0)}
(n) \triangleq n$ and if $\log^{(i)} (n)=0$, we define
$\log^{(j)} (n) = 0$, for every $j>i$.
For $n \geq 1$, define $\log^{*} (n)\triangleq \min \{i:
\log^{(i)}(n) \leq 1\}$.

A subset $I\subseteq V$ is an \emph{independent set} if no two vertices in $I$ are an
edge in $E$. An independent set $I$ is \emph{maximal} if $I\cup\{v\}$ is not an
independent set for every $v\in V\setminus I$. We use \mis\ as an abbreviation of a
maximal independent set.

A subset $M\subseteq E$ is a matching if no two edges in $M$ share an endpoint.  Let
$M^*$ denote a maximum cardinality matching of $G$.  We say that a matching $M$ is a
$(1-\eps)$-approximate maximum matching if
\[
    |M| \geq (1-\eps)\cdot|M^*|\:.
\]
Let $w(e)$ denote the weight of an edge $e\in E$. The weight of a subset $F\subseteq
E$ is $\sum_{e\in F} w(e)$ and is denoted by $w(F)$.  Let $M_w^*$ denote a maximum
weight matching of $G$. A matching $M$ is a $(1-\eps)$-approximate maximum weight matching
if $w(M) \geq (1-\eps)\cdot w(M_w^*)$. We abbreviate the terms maximum cardinality
matching and maximum weight matching by \mcm\ and \mwm, respectively.

\subsection{The \clocal\ Model}\label{sec:cmodel}
The model of centralized local computations was defined in~\cite{shaiics2011}.
In this section we describe this model for problems over labeled graphs.

\paragraph{Labeled graphs.}
An undirected graph $G=(V,E)$ is labeled if: (1)~Vertex names are distinct and each
have description of at most $O(\log n)$ bits. For
simplicity, assume that the vertex names are in $\{1,\ldots,n\}$. We denote the
vertex whose name is $i$ by $v_i$. (2)~Each vertex $v$ holds a list of $\degree(v)$
pointers, called \emph{ports}, that point to the neighbors of $v$.  The assignment
of ports to neighbors is arbitrary and fixed.

\paragraph{Problems over labeled graphs.}Let $\Pi$ denote a computational problem
over labeled graphs (e.g., maximum matching, maximal independent set, vertex
coloring).  A solution for problem $\Pi$ over a labeled graph $G$ is a function, the
domain and range of which depend on $\Pi$ and $G$.  For example: (1)~In the Maximal
Matching problem, a solution is an indicator function $M:E\rightarrow \{0,1\}$ of a
maximal matching in $G$. (2)~In the problem of coloring the vertices of a graph by
$(\Delta+1)$ colors, a solution is a coloring $c:V\rightarrow \{1,\dots,\Delta+1\}$.
Let $\sol(G,\Pi)$ denote the set of solutions of problem $\Pi$ over the labeled graph
$G$.

\paragraph{Probes.}
In the \clocal\ model, access to the labeled graph is limited to probes.  A
\emph{probe} is a pair $(v,i)$ that asks ``who is the $i$th neighbor of $v$?''.  The
answer to a probe $(v,i)$ is as follows.  (1)~If $\degree(v)<i$, then the answer is
``null''. (2)~If $\degree(v)\geq i$, then the answer is the (ID of) vertex $u$ that is
pointed to by the $i$th port of $v$. For simplicity, we assume that the answer also
contains the port number $j$ such that $v$ is the $j$th neighbor of $u$. (This
assumption reduces the number of probes by at most a factor of $\Delta$.)

\paragraph{Online Property of \clocal-algorithms.}
The input of an algorithm $\alg$ for a problem $\Pi$ over labeled graphs in the
\clocal\ model consists of three parts: (1)~access to a labeled graph $G$ via probes,
(2)~the number of vertices $n$ and the maximum degree $\Delta$ of the graph $G$, and
(3)~a sequence $\{q_i\}_{i=1}^N$ of queries.  Each query $q_i$ is a request for an
evaluation of $f(q_i)$ where $f\in \sol(G,\Pi)$.  Let $y_i$ denote the output of
\alg\ to query $q_i$. We view algorithm \alg\ as an online algorithm because it must output
$y_i$ without any knowledge of subsequent queries.

\paragraph{Consistency.} We say that $\alg$ is \emph{consistent with $(G,\Pi)$} if
\begin{equation}\label{eq:const-def}
    \exists f \in  \sol(G,\Pi) \mbox{ s.t. }~\forall N \in \NN ~~\forall \{q_i\}_{i=1}^N ~~\forall i ~:~y_i = f(q_i)\:.
\end{equation}

\paragraph{Examples.} Consider the problem of computing a maximal independent set.
The \clocal-algorithm is input a sequence of queries $\{q_i\}_i$, each of which is a
vertex. For each $q_i$, the algorithm outputs whether $q_i$ is in $I$, for an
arbitrary yet fixed maximal independent set $I\subseteq V$.  Consistency means that
$I$ is fixed for all queries. The algorithm has to satisfy this specification even
though it does not probe all of $G$, and obviously does not store the maximal
independent set $I$.  Moreover, a stateless algorithm does not even remember the
answers it gave to previous queries.  Note that if a vertex is queried twice, then
the algorithm must return the same answer.  If two queries are neighbors, then the
algorithm may not answer that both are in the independent set. If the algorithms
answers that $q_i$ is not in the independent set, then there must exist a neighbor of
$q_i$ for which the algorithm would answer affirmatively. If all vertices are
queried, then the answers constitute the maximal independent set $I$.

As another example, consider the problem of computing a $(\Delta+1)$ vertex coloring.
Consistency in this example means the following. The online algorithm is input a
sequence of queries, each of which is a vertex. The algorithm must output the color
of each queried vertex. If a vertex is queried twice, then the algorithm must return
the same color. Moreover, queried vertices that are neighbors must be colored by
different colors.  Thus, if all vertices are queried, then the answers constitute a
legal vertex coloring that uses $(\Delta+1)$ colors.

\paragraph{Resources and Performance Measures.} The resources used by a
\clocal-algorithm are: probes, space, and random bits.  The running time used to
answer a query is not counted.  The main performance measure is the \emph{maximum
  number of probes} that the \clocal-algorithm performs per query.  We consider an
additional measure called \emph{probe radius}.  The probe radius of a
\clocal-algorithm $C$ is $r$ if, for every query $q$, all the probes that algorithm
$C$ performs in $G$ are contained in the ball of radius $r$ centered at $q$.
We denote the probe radius of algorithm $A$ over graph $G$ by $r_G(A)$.

The \emph{state} of algorithm $\alg$ is the information that $\alg$ saves between
queries.  The \emph{space} of algorithm $\alg$ is the maximum number of bits required
to encode the state of $\alg$.  A \clocal-algorithm is \emph{stateless} if the
algorithm does not store any information between queries.  In particular, a stateless
algorithm does not store previous queries, answers to previous probes, or answers
given to previous queries.\footnote{\label{foot:4}We remark that in~\cite{shaiics2011} no
  distinction was made between the space needed to answer a query and the space
  needed to store the state between queries. Our approach is different and follows
  the \dlocal\ model in which one does not count the space and running time of the
  vertices during the execution of the distributed algorithm.  Hence, we ignore the
  space and running time of the \clocal-algorithm during the processing of a query.
  Interestingly, the state between queries
  in~\cite{shaisoda2012,shaiicalp2012,shaiapprox2013,DBLP:journals/corr/ReingoldV14}
  only stores a random seed that is fixed throughout the execution of the
  algorithm.}  In this paper all our \clocal-algorithms are stateless.

\begin{definition}
An online algorithm is a $\clocal[q,s]$ algorithm for $\Pi$ if
  \begin{inparaenum}[(1)]
  \item it is consistent with $(G,\Pi)$,
  \item it performs at most $q$ probes, and
  \item the space of the algorithm is bounded by $s$.
  \end{inparaenum}
\end{definition}
\sloppy The goal in designing algorithms in the \clocal\ model is to minimize the
number of probes and the space (in particular $q,s = o(n)$). A $\clocal[q,s]$
algorithm with $s=0$ is called a stateless $\clocal[q]$ algorithm.  Stateless
algorithms are useful in the case of uncoordinated distributed servers that answer
queries without communicating with each other.

\paragraph{Randomized local algorithms.}
A randomized $\clocal$-algorithm is also parameterized by the \emph{failure
  probability} $\delta$.  We say that $\alg$ is a $\clocal[q,s,\delta]$ algorithm for
$\Pi$ if the algorithm is consistent, performs at most $q$ probes, and uses space at
most $s$ with probability at least $1-\delta$.  The standard requirement is that
$\delta = 1/\poly(n)$.

The number of random bits used by a randomized algorithm is also a resource.  One can
distinguish between two types of random bits: (1)~random bits that the algorithm must
store between queries, and (2)~random bits that are not stored between queries.  We
use the convention that information that is stored between queries is part of the
state. Hence, random bits, even though chosen before the first query, are included in
the state if they are stored between queries.\footnote{As noted in Footnote~\ref{foot:4},
  in~\cite{shaisoda2012,shaiicalp2012,shaiapprox2013} the state does not change
  during the execution of the \clocal\ algorithm.}

\paragraph{Parallelizability and query order obliviousness.}
In~\cite{shaisoda2012,shaiicalp2012,shaiapprox2013} two requirements are introduced:
\emph{parallelizability} and \emph{query order obliviousness}. These requirements are
fully captured by the definition of a consistent, online, deterministic algorithm
with zero space.
That is, every online algorithm that is consistent, stateless, and deterministic is
both parallelizable and query order oblivious.

\subsection{The \dlocal\ Model}\label{sec:dlocal-def}
The model of local distributed computation is a classical model
(e.g.,~\cite{linial,pelegbook,dsurvey}). A distributed computation takes place in an
undirected labeled graph $G=(V,E)$. The neighbors of each vertex $v$ are numbered
from $1$ to $\deg(v)$ in an arbitrary but fixed manner. Ports are used to point to
the neighbors of $v$; the $i$th port points to the $i$th neighbor.  Each vertex
models a processor, and communication is possible only between neighboring
processors.  Initially, every $v \in V$ is input a local input.  The computation is
done in $r \in \NN$ synchronous rounds as follows.  In every round: (1)~every
processor receives a message from each neighbor, (2)~every processor performs a
computation based on its local input and the messages received from its neighbors,
(3)~every processor sends a message to each neighbor.  We assume that a message sent
in the end of round $i$ is received in the beginning of round $i+1$.  After the $r$th
round, every processor computes a local output.

The following assumptions are made in the \dlocal\ model:
\begin{inparaenum}[(1)]
\item The local input to each vertex $v$ includes the ID of $v$, the degree of the
  vertex $v$, the maximum degree $\Delta$, the number of vertices $n$, and the ports
  of $v$ to its neighbors.
\item The IDs are distinct and bounded by a polynomial in $n$.
\item The length of the messages sent in each round is not bounded.
\item The computation in each vertex in each round need not be efficient.
\end{inparaenum}

We say that a distributed algorithm is a \emph{$\dlocal[r]$-algorithm} if the number
of communication rounds is $r$.  Strictly speaking, a distributed algorithm is
considered {\em local\/} if $r$ is bounded by a constant.  We say that a
$\dlocal[r]$-algorithm is \emph{almost local} if $r=O(\log^*(n))$.  When it is
obvious from the context we refer to an almost \dlocal\ algorithm simply by a
\dlocal\ algorithm.

\subsection{Mutual Simulations Between \dlocal\ and \clocal}\label{sec:mutual}
In this section we show that one can simulate algorithms over labeled graphs in one
model by algorithms in the other model (without any restriction on $\Delta$).
Since our algorithms are deterministic, we focus on simulations of deterministic algorithms.

The following definition considers \clocal-algorithms whose queries are vertices of a
graph. The definition can be easily extended to edge queries.
\begin{definition}
  A \clocal-algorithm $C$ simulates (or is simulated by) a \dlocal-algorithm $D$ if,
  for every vertex $v$, the local output of $D$ in vertex $v$ equals the answer that
  algorithm $C$ computes for the query $v$.
\end{definition}
\paragraph{Simulation of \dlocal\ by \clocal~\cite{parnasron}:}
Every deterministic $\dlocal[r]$-algorithm, can be simulated by a deterministic,
stateless \clocal$[O(\Delta^r)]$-algorithm. The simulation proceeds simply by probing
all vertices in the ball of radius $r$ centered at the query.  If $\Delta=2$, then
balls are simple paths (or cycles) and hence simulation of a $\dlocal[r]$-algorithm
is possible by a \clocal$[2r]$-algorithm.

\paragraph{Simulation of \clocal\ by \dlocal:}
The following Proposition suggests a design methodology for distributed algorithms.
For example, suppose that we wish to design a distributed algorithm for maximum
matching. We begin by designing a \clocal-algorithm $C$ for computing a maximum
matching. Let $r$ denote the probe radius of the \clocal-algorithm $C$. The
proposition tells us that we can compute the same matching (that is computed by $C$)
by a distributed $r$-round algorithm.%
\footnote{Message lengths grow at
a rate of $O(\Delta^{r+1} \cdot \log n)$ as information (e.g., IDs and existence of
edges) is accumulated.}%
\begin{proposition}\label{prop:simul}
  Every stateless deterministic \clocal-algorithm $C$ whose probe radius is at most $r$ can be
  simulated by a deterministic $\dlocal[r]$-algorithm $D$.
\end{proposition}
\begin{proof}
  The distributed algorithm $D$ collects, for every $v$, all the information in the
  ball of radius $r$ centered at $v$. (This information includes the IDs of the
  vertices in the ball and the edges between them.)

  After this information is collected, the vertex $v$ locally runs the
  \clocal-algorithm $C$ with the query $v$. Because algorithm $C$ is stateless, the
  vertex has all the information required to answer every probe of $C$.
\end{proof}

\section{Acyclic Orientation}\label{sec:obr}
In this section we deal with orientation of undirected graphs, namely, assigning
directions to the edges. We suggest to obtain an orientation from a vertex coloring.

\paragraph{Definitions.}
An \emph{orientation} of an undirected graph $G=(V,E)$ is a
directed graph $H=(V,A)$, where $\{u,v\}\in E$ if and only
if $(u,v)\in A$ or $(v,u)\in A$ but not both. An
orientation $H$ is \emph{acyclic} if there are no directed
closed paths in $H$. The \emph{radius} of an acyclic
orientation $H$ is the length of the longest directed path
in $H$. We denote the radius of an orientation by
$\rad(H)$.  In the problem of acyclic orientation with
bounded radius (\orad), the input is an undirected graph.
The output is an orientation $H$ of $G$ that is acyclic.
The goal is to compute an acyclic orientation $H$ of $G$
that minimizes $\rad(H)$.

The set of vertices that are reachable from $v$ in a directed graph $H$ is called the
\emph{reachability set} of $v$.  We denote the reachability set of $v \in V$ in
digraph $H$ by $\Reach_H(v)$. Let $\reach_H(v) \triangleq |\Reach_H(v)|$ and
$\reach(H) \triangleq\max_{v \in V}\reach_H(v)$.  We simply write
$\Reach(v),\reach(v)$ when the digraph $H$ is obvious from the
context. In the problem of acyclic orientation with bounded reachability (\obr), the input is
an undirected graph. The output is an orientation $H$ of $G$ that is acyclic. The
goal is to minimize $\reach(H)$.

Previous works obtain an acyclic orientation by random vertex
ranking~\cite{onak2008,yoshida2009improved,shaisoda2012,shaiicalp2012,shaiapprox2013}.
We propose to obtain an acyclic orientation by vertex coloring.

\begin{proposition}[Orientation via coloring]\label{prop:color2ori}
  Every coloring by $c$ colors induces an acyclic orientation with
 \[
    \rad(H) \leq c-1.
  \]
  Hence, every $\clocal[q]$-algorithm for vertex coloring also implies a
  $\clocal[2q]$-algorithm for acyclic orientation.
\end{proposition}
\begin{proof}
  Direct each edge from a high color to a low color.  By monotonicity the orientation
  is acyclic.  Every directed path has at most $c$ vertices, and hence the
  reachability is bounded as required. To determine the orientation of an edge
  $(u,v)$, the \clocal-algorithm simply computes the colors of $u$ and $v$.
\end{proof}

The following proposition bounds the maximum cardinality of a reachability by a
function of the reachability radius.
\begin{proposition}\label{prop:rad2reach}
  \[
    \reach(H) \leq 1+\Delta\cdot\sum_{i=1}^{\rad(H)}(\Delta-1)^{i-1}\leq\begin{cases}
      2\Delta\cdot(\Delta-1)^{\rad(H)-1}, & \text{if }\Delta \geq 3,\\
      2\cdot \rad(H)+1, & \text{if }\Delta=2\:.
    \end{cases}
  \]
\end{proposition}
\subsection{A \clocal\ Algorithm for Vertex Coloring}

In this section we present a deterministic, stateless
\clocal$[O(\Delta\cdot \log^* n +\Delta^3)]$-algorithm that
computes a vertex coloring that uses $c=O(\Delta^2)$ colors
(see Theorem~\ref{thm:colorpoly}). Orientation by this
coloring yields an acyclic orientation $H$ with $\rad(H)
\leq \Delta^2$ and $\reach(H)\leq 2\cdot \Delta^c$.

$\clocal$-algorithms for vertex coloring can be also obtained by simulating \dlocal\
vertex coloring algorithms. Consider, for example, the $(\Delta+1)$ coloring using
$r_1=O(\Delta)+\frac 12 \cdot \log^* n$ rounds of ~\cite{barenboim2009distributed} or
the $O(\Delta^2)$ coloring using $r_2= O(\log^* n)$ rounds of~\cite{linial}.
\clocal\ simulations of these algorithms require $O(\Delta^{r_i})$ probes. Thus, in
our algorithm, the number of probes grows (slightly) slower as a function of $n$ and
is polynomial in $\Delta$.

Our algorithm relies on techniques from two previous \dlocal\ coloring algorithms.
\begin{theorem}[{\cite[Corollary~4.1]{linial}}]\label{thm:colorred}
A  $5\Delta^2\log c$ coloring can be computed from a $c$ coloring by a \dlocal$[1]$-algorithm.
\end{theorem}
\begin{lemma}[Linial 92,Lemma~4.2]\label{thm:colorred2}
A  $O(\Delta^2)$-coloring can be computed from a $O(\Delta^3)$-coloring by a
\dlocal$[1]$-algorithm.
\end{lemma}
\begin{theorem}[{\cite[Section~4]{panconesi2001some}}]\label{cv}
A $(\Delta+1)$ coloring can be computed by a \dlocal$[O(\Delta^2+\log^*n)]$-algorithm.
\end{theorem}

\begin{theorem}\label{thm:colorpoly}
 An $O(\Delta^2)$ coloring
 can be computed by a deterministic, stateless \clocal$[O(\Delta \cdot \log^*
 n+\Delta^3)]$-algorithm. The probe radius of this algorithm is $O(\log^* n)$.
\end{theorem}

\begin{proof}
  We begin by describing a two phased \dlocal$[O(\log^*n)]$-algorithm $D$ that uses
  $O(\Delta^2)$ colors. Algorithm $D$ is especially designed so
  that it admits an ``efficient'' simulation by a \clocal-algorithm.

  Consider a graph $G=(V,E)$ with a maximum degree $\Delta$.  In the first phase, the
  edges {\red are} partitioned into $\Delta^2$ parts, so that the maximum degree in each part is
  at most $2$. Let $p_i(u)$ denote the neighbor of vertex $u$ pointed to by the $i$th
  port of $u$.
  {\red Following Kuhn~\cite{kuhn2009weak} we partition the edge set $E$ as follows.}
  Let $E_{\{i,j\}}\subseteq E$ be defined by
\[
E_{\{i,j\}} \triangleq \{ \{u,v\} \mid p_i(u)=v, p_j(v)=u\}.
\]
Each edge belongs to exactly one part $E_{\{i,j\}}$. For each part $E_{\{i,j\}}$ and vertex
$u$, at most two edges in $E_{\{i,j\}}$ are incident to $u$. Hence, the maximum degree in
each part is at most $2$. Each vertex can determine in a single round how the edges incident to it are
partitioned among the parts. Let $G_{\{i,j\}}$ denote the undirected graph over $V$ with edge set $E_{\{i,j\}}$.

By Theorem~\ref{cv}, we $3$-color each graph $G_{\{i,j\}}$
in $O(\log^* n)$ rounds. This induces a vector of
$\Delta^2$ colors per vertex, hence a $3^{\Delta^2}$ vertex
coloring of $G$.

In the second phase, Algorithm $D$ applies Theorem~\ref{thm:colorred} twice, followed by an application of Theorem~\ref{thm:colorred2}, to reduce the number of colors to $O(\Delta^2)$.

We now present an efficient simulation of algorithm $D$ by a \clocal-algorithm $C$.
Given a query for the color of vertex $v$, Algorithm $C$ simulates the first phase of
$D$ in which a $3$-coloring algorithm is executed in each part $E_{\{i,j\}}$.  Since the
maximum degree of each $G_{\{i,j\}}$ is two, a ball of radius $r$ in $G_{\{i,j\}}$ contains
at most $2r$ edges. In fact, this ball can be recovered by at most $2r$ probes. It
follows that a \clocal\ simulation of the $3$-coloring of $G_{\{i,j\}}$ performs only
$O(\log^* n)$ probes. Observe that if vertex $v$ is isolated in $G_{\{i,j\}}$, then it
may be colored arbitrarily (say, by the first color). A vertex $v$ is not isolated in
at most $\Delta$ parts. It follows that the simulation of the first phase performs
$O(\Delta \cdot \log^* n)$ probes.

The second phase of algorithm $D$ requires an additional $\Delta^3$ probes, and the
theorem follows.
\end{proof}

The following corollary is a direct consequence of the coloring
algorithm described in Theorem~\ref{thm:colorpoly}, the
orientation induced by the coloring in
Proposition~\ref{prop:color2ori}, and the bound on the
reachability based on the radius in
Proposition~\ref{prop:rad2reach}.
\begin{coro}
  \label{coro:improve}
  There is a deterministic, stateless \clocal$[O(\Delta\cdot \log^*
  n+\Delta^3)]$-algorithm for orienting a graph that achieves $\rad(H)\leq \Delta^2$
  and $\reach(H) \leq \Delta^{O(\Delta^2)}$.
\end{coro}

\section{Deterministic Localization of Sequential Algorithms and Applications}\label{sec:lin}

A common theme in online algorithms and ``greedy'' algorithms is that the elements
are scanned in query order or in an arbitrary order, and a decision is made for each
element based on the decisions of the previous elements. Classical examples of such
algorithms include the greedy algorithms for maximal matchings, $(\Delta+1)$ vertex
coloring, and maximal independent set. We present a compact and axiomatic \clocal\
deterministic simulation of this family of algorithms, for which a randomized
simulation appeared in~\cite{shaiicalp2012}.  Our deterministic simulation is based
on an acyclic orientation that induces a partial order.

For simplicity, consider a graph problem $\Pi$, the solution of which is a function
$g(v)$ defined over the vertices of the input graph. For example, $g(v)$ can be the
color of $v$ or a bit indicating if $v$ belongs to a maximal independent set. (One
can easily extend the definition to problems in which the solution is a function over
the edges, e.g., maximal matching.)

We refer to an algorithm as a \emph{sequential algorithm} %
if it fits the scheme listed as Algorithm~\ref{alg:greedy}.  The algorithm
$\alg(G,\sigma)$ is input a graph $G=(V,E)$ and a bijection $\sigma:\{1,\ldots,n\}
\rightarrow V$ of the vertices.  The bijection $\sigma$ orders the vertices in total
order, if $\sigma(i)=v$ then $v$ is the $i$th vertex in the order and
$\sigma^{-1}(v)=i$.  The algorithm scans the vertices in the order induced by
$\sigma$. It determines the value of $g(\sigma(i))$ based on the values of its
neighbors whose value has already been determined. This decision is captured by %
the function $f$ in Line~\ref{line:f}.  For example, in vertex coloring, $f$ returns
the smallest color that 
does not appear in a given a subset of colors.

\begin{algorithm}
  \caption{The sequential algorithm scheme.}
  \label{alg:greedy}
  \begin{algorithmic}[1]
  \Require A graph $G=(V,E)$ and a bijection $\sigma:\{1,\ldots,n\} \rightarrow V$.
  \For {$i=1$ to $n$}
    \State $g({\sigma(i)}) \gets f\left(\left\{g(v) : v\in \Gamma(\sigma(i)) \;\&\; \sigma^{-1}(v)<i\right\}\right)$
\label{line:f}
    \Comment{(Decide based on ``previous'' neighbors)}
  \EndFor
  \State Output:  $g$.
  \end{algorithmic}
\end{algorithm}

\begin{lemma}\label{lemma:linear}
Let $G= (V,E)$ be a graph, let $H=(V,A)$ be an acyclic orientation of $G$
and let $P_{>} \subseteq V\times V$ denote the partial order defined by the transitive closure of $H$.
Namely, $(u,v) \in P_{>}$ if and only if there exists a directed path from $u$ to $v$ in $H$.
   Let $\alg$ denote a sequential
  algorithm.  For every bijection $\sigma:\{1,\dots,n\}\to V$ that is a linear extension of $P_{>}$
  (i.e, for every $(u,v) \in P_{>}$ we have
  that $\sigma^{-1}(u) > \sigma^{-1}(v)$),
  the output of $\alg(G,\sigma)$ is the same.
\end{lemma}
\begin{proof}
  Consider two linear extensions $\sigma$ and $\tau$ of $P_{>}$.
  Let $g_{\sigma}$ denote the
  output of $\alg(G,\sigma)$ and define $g_\tau$ analogously.

  Let\[ B_\sigma(u) \eqdf \{v \in \Gamma(u) \mid \sigma^{-1}(u) > \sigma^{-1}(v)\}\:.
  \]
  We claim that $B_\sigma(u)=B_\tau(u)$ for every $u$. By symmetry, it suffices to
  prove that $B_\sigma(u)\subseteq B_\tau(u)$. Consider a vertex $v\in
  B_{\sigma}(u)$. We need to show that $v\in B_{\tau}(u)$. By definition, $v$ is a
  neighbor of $u$. We consider the two possible orientations of the edge $(u,v)$. If
  $(u,v)\in A$, then $(u,v)\in P_{>}$.  Hence $\sigma^{-1}(u) > \sigma^{-1}(v)$ and
  $\tau^{-1}(u) > \tau^{-1}(v)$ because $\sigma$ and $\tau$ are linear extensions of
  $P_{>}$. We conclude that $v\in B_\tau(u)$, as required. If $(v,u)\in A$, then
  $\sigma^{-1}(u) < \sigma^{-1}(v)$, and this implies that $v\not\in B_\sigma(u)$, a
  contradiction.

  To complete the proof, we prove by induction on $i$ that for $u=\sigma^{-1}(i)$ we
  have $g_\sigma(u)=g_\tau(u)$.  Indeed, $g_\sigma(u)=f(B_\sigma(u))$ and
  $g_\tau(u)=f(B_\tau(u))$.  For $i=1$ we have $B_\sigma(u)=B_\tau(u)=\emptyset$,
  hence $f(B_\sigma(u)) =f(B_\tau(u))$, as required.  To induction step recall that
  $B_\sigma(u)=B_\tau(u)$. By the induction hypothesis we conclude that
  $f(B_\sigma(u)) =f(B_\tau(u))$, and the lemma follows.
\end{proof}

The following theorem proves that a sequential algorithm can be simulated by a
$\clocal[q]$-algorithm. The number of probes $q$ equals the number of probes used by
the vertex coloring algorithm (that induces an acyclic orientation) times the
max-reachability of the orientation.
\begin{theorem}\label{thm:reduce}
  For every sequential algorithm $\alg$, there exists a deterministic, stateless
  \clocal$[\Delta^{O(\Delta^2)}\cdot \log^* n]$-algorithm $\alg_c$ that
  simulates $\alg$ in the following sense.  For every graph $G$, there exists a bijection
  $\sigma$, such that $\alg_c(G)$ simulates $\alg(G,\sigma)$.
  That is, for every vertex $v$ in $G$, the answer of $\alg_c(G)$ on query $v$
  is $g_\sigma(v)$, where  $g_{\sigma}$ denotes the
  output of $\alg(G,\sigma)$.
\end{theorem}
\begin{proof}
  Consider the acyclic orientation $H$ of $G$ computed by the \clocal$[\Delta\cdot
  \log^* n+\Delta^3]$-algorithm presented in Corollary~\ref{coro:improve}.  Let
  $P_{>}$ denote the partial order that is induced by $H$, and let $\sigma$ be any
  linear extension of $P_{>}$ (as defined in Lemma~\ref{lemma:linear}).  On query $v
  \in V$ the value $g_\sigma(v)$ is computed by performing a (directed) DFS on $H$
  that traverses the subgraph of $H$ induced by $R_H(v)$. The DFS uses the \clocal\
  algorithm from Corollary~\ref{coro:improve} to determine the orientation of each
  incident edge and continues only along outward-directed edges\footnote{Given that
    the \clocal\ algorithm works by running a \clocal\ coloring algorithm, one can
    actually use the latter algorithm directly.}.  The value of $g_\sigma(v)$ is
  determined when the DFS backtracks from $v$. The product of
  $\reach(H)=\Delta^{O(\Delta^2)}$ and the number of probes of the
  orientation algorithm bounds the number of probes of $\alg_c$. Hence, we obtain
  that $\Delta^{O(\Delta^2)}\cdot \log^* n$ probes suffice, and the
  theorem follows.
\end{proof}

\begin{coro}\label{coro:mis}
  There are deterministic, stateless \clocal$[\Delta^{O(\Delta^2)} \cdot
  \log^* n]$ algorithms for $(\Delta+1)$-vertex coloring, maximal independent set,
  and maximal matching.
\end{coro}

We have described two $\clocal$ coloring algorithms; one
uses $\Delta^2$ colors (Theorem~\ref{thm:colorpoly}), and
the second uses $\Delta+1$ colors
(Corollary~\ref{coro:mis}). The number of probes of the
$(\Delta+1)$-coloring obtained by simulating the sequential
coloring algorithm is exponential in $\Delta$. The
$\Delta^2$-coloring algorithm requires only
$O(\Delta\cdot\log^*n+\Delta^3)$ probes. Hence, increasing
the number of colors (by a factor of $\Delta$) enables us
to reduce the dependency of the number of probes on the
maximum degree.

We conclude with the following immediate lemma that bounds the probe radius
of the \clocal-algorithm for \mis.
\begin{lemma}\label{lemma:simul}
  Let $\ao$ denote a stateless \clocal-algorithm that computes an acyclic orientation
  $H=(V,A)$ of a graph $G=(V,E)$.  Let $r$ denote the probe radius of $\ao$.
  Then, there exists a stateless \clocal-algorithm for \mis\ whose probe radius is at most
  $r+\rad(H)$.
\end{lemma}
Assume that the acyclic orientation is based on the
$\clocal[O(\Delta\cdot\log^*n+\Delta^3)]$-algorithm that computes a $\Delta^2$-vertex
coloring. The probe radius of the \mis-algorithm implied by lemma~\ref{lemma:simul}
is $O(\log^*n+\Delta^2)$.  Indeed, the probes of the $\Delta^2$-coloring algorithm
are confined to a ball of radius $O(\log^* n)$.  The probes of the simulation of the
sequential algorithm are confined to a ball of radius $c=O(\Delta^2)$.

Let \lmis\ denote the \clocal\-algorithm for maximal
independent set (\mis) stated in Corollary~\ref{coro:mis}.
The Boolean predicate $\lmis(G,v)$ indicates if $v$ is in
the \mis\ of $G$ computed by Algorithm \lmis.
\section{A  \clocal\ Approximate \mcm\ Algorithm}\label{sec:clocal mcm}
In this section we present a stateless deterministic
\clocal\ algorithm that computes a $(1-\eps)$-approximation
of a maximum cardinality matching. The algorithm is based
on a \clocal-algorithm for maximal independent set (see
Corollary~\ref{coro:mis}) and on the local improvement
technique of Nguyen and Onak~\cite{onak2008}.

\paragraph{Terminology and Notation.}
Let $M$ be a matching in $G=(V,E)$.
A vertex $v \in V$ is \emph{$M$-free} if $v$ is not an endpoint of an edge in $M$.
A simple path is \emph{$M$-alternating} if it consists of edges drawn alternately from $M$ and from $E \setminus M$.
A path is \emph{$M$-augmenting} if it is $M$-alternating and if both of the path's endpoints are $M$-free vertices.
Note that the length of an augmenting path must be odd.
The set of edges in a path $p$ is denoted by $E(p)$, and the set of edges in a
collection $P$ of paths is denoted by $E(P)$.
Let $A\oplus B$ denote the symmetric difference of the sets $A$ and $B$.

\paragraph{Description of The Global Algorithm.}
Similarly to~\cite{DBLP:conf/spaa/LotkerPP08,onak2008, shaiapprox2013} our local
algorithm simulates the global algorithm listed as Algorithm~\ref{alg:global mcm}. This
global algorithm builds on  lemmas of Hopcroft and
Karp~\cite{hopcroft1973n} and Nguyen and Onak~\cite{onak2008}.

\begin{lemma}[\cite{hopcroft1973n}]
  Let $M$ be a matching in a graph $G$. Let $k$ denote the length of a shortest
  $M$-augmenting path. Let $P^*$ be a maximal set of vertex disjoint $M$-augmenting
  paths of length $k$.  Then, $(M\oplus E(P^*))$ is a matching and the length of every $(M\oplus E(P^*))$-augmenting path
  is at least $k+2$.
  \label{lemma:hopcroft-karp}
\end{lemma}

\begin{lemma}[{\cite[Lemma~6]{onak2008}}]\label{lemma:apxmatch}
  Let $M^*$ be a maximum matching and $M$ be a matching in a graph $G$. Let $2k+1$
  denote the length of a shortest $M$-augmenting path.  Then
  \[
    |M| \geq \frac{k}{k+1}\cdot |M^*|\:.
  \]
\end{lemma}

\begin{algorithm}[H]
\caption{$\text{Global-APX-MCM}(G,\eps)$.}
\label{alg:global mcm}
\begin{algorithmic}[1]
\Require A graph $G=(V,E)$ and $0<\eps <1$.
\Ensure  A $(1-\eps)$-approximate matching
\State $M_0\gets \emptyset$.
\State $k\gets \lceil \frac{1}{\eps} \rceil$.
\For {$i=0$ to $k$}
  \State $P_{i+1} \gets \{ p \mid \text{$p$ is an
      $M_{i}$-augmenting path}, |p|=2i+1\}$.
  \State $P^*_{i+1}\subseteq P_{i+1}$ is a maximal vertex
      disjoint subset of paths.
  \State $M_{i+1}\triangleq M_{i} \oplus E(P^*_{i+1})$.
  \EndFor
\State \textbf{Return} $M_{k+1}$.
\end{algorithmic}
\end{algorithm}
%
%
\begin{algorithm}[H]
\caption{$\text{Global-APX-MCM'}(G,\eps)$.}
\label{alg:global mcm'}
\begin{algorithmic}[1]
\Require A graph $G=(V,E)$ and $0<\eps <1$.
\Ensure  A $(1-\eps)$-approximate matching
\State $M_0\gets \emptyset$.
\State $k\gets \lceil \frac{1}{\eps} \rceil$.
\For {$i=0$ to $k$}
  \State Construct the intersection graph $H_i$ over $P_i$.
  \State $P^*_{i+1} \gets \mis (H_i)$.
  \State $M_{i+1}\triangleq M_{i} \oplus E(P^*_{i+1})$.
  \EndFor
\State \textbf{Return} $M_{k+1}$.
\end{algorithmic}
\end{algorithm}

Algorithm~\ref{alg:global mcm} is given as input a graph $G$ and an approximation parameter
$\eps\in (0,1)$. The algorithm works in $k+1$
iterations, where $k=\lceil \frac{1}{\eps}\rceil$.  Initially, $M_0=\emptyset$. The
invariant of the algorithm is that $M_i$ is a matching, every augmenting path of
which has length at least $2i+1$.  Given $M_{i}$, a new matching $M_{i+1}$ is
computed as follows.  Let $P_{i+1}$ denote the set of shortest $M_{i}$-augmenting
paths. Let $P^*_{i+1} \subseteq P_{i+1}$ denote a maximal subset of vertex disjoint
paths.  Define $M_{i+1}\triangleq M_{i} \oplus E(P^*_{i+1})$.  By
Lemmas~\ref{lemma:hopcroft-karp} and~\ref{lemma:apxmatch}, we obtain the following
result.

\begin{theorem}
  The matching $M_{k+1}$ computed by Algorithm~\ref{alg:global mcm} is a
  $(1-\eps)$-approximation of a maximum matching.
\end{theorem}

\paragraph{The intersection graph.}
Define the intersection graph $H_i=(P_i,C_i)$ as follows.  The set of nodes $P_i$ is
the set of $M_{i-1}$-augmenting paths of length $2i-1$.  We connect two paths $p$ and
$q$ in $P_i$ by an edge $(p,q) \in C_i$ if $p$ and $q$ intersect (i.e., share a
vertex in $V$).  Note that $H_1$ is the line graph of $G$ and that $M_1$ is simply a
maximal matching in $G$.  Observe that $P^*_i$ as defined above is a maximal
independent set in $H_i$.  Thus, iteration $i$ of the global algorithm can be
conceptualized by the following steps (see Algorithm~\ref{alg:global mcm'}): construct the intersection graph $H_i$,
compute a maximal independent set $P^*_i$ in $H_i$, and augment the matching by
$M_{i}\triangleq M_{i-1} \oplus(E(P^*_i))$.

%

\paragraph{Implementation by a stateless deterministic \clocal\ Algorithm.}
The recursive local improvement technique in~\cite[Section
3.3]{onak2008} simulates the global algorithm. It is based
on a recursive oracle $\oracle_i$.  The input to oracle
$\oracle_i$ is an edge $e \in E$, and the output is a bit
that indicates whether $e \in M_i$. Oracle $\oracle_i$
proceeds by computing two bits $\tau$ and $\rho$ (see
Algorithm~\ref{alg:Oi}).  The bit $\tau$ indicates whether
$e\in M_{i-1}$, and is computed by invoking oracle
$\oracle_{i-1}$.  The bit $\rho$ indicates whether $e\in
E(P^*_i)$ (where $P^*_i$ is an \mis\ in $H_{i-1}$).  Oracle
$\oracle_i$ returns $\tau \oplus \rho$ because
$M_i=M_{i-1}\oplus E(P^*_i)$.

We determine whether $e\in E(P^*_i)$ by running the \clocal-algorithm $\proca_i$ over
$H_i$ (see Algorithm~\ref{alg:Ai}).  Note that $\proca_1$ simply computes a maximal
matching (i.e., a maximal independent set of the line graph $H_1$ of $G$).  The main
difficulty we need to address is how to simulate the construction of $H_i$ and probes
to vertices in $H_i$. We answer the question whether $e \in E(P_i^*)$ by executing
the following steps: (1)~Listing: construct the set $P_i(e)\triangleq \{p\in P_i \mid
e\in E(p)\}$. Note that $e\in E(P^*_i)$ if and only if $P_i(e)\cap P^*_i\neq
\emptyset$. (2)~\mis-step: for each $p\in P_i(e)$, input the query $p$ to an
\mis-algorithm for $H_i$ to test whether $p\in P^*_i$. If an affirmative answer is
given to one of these queries, then we conclude that $e\in E(P_i^*)$. We now
elaborate on how the listing step and the \mis-step are carried out by a
\clocal-algorithm.

The listing of all the paths in $P_i(e)$ uses two preprocessing steps: (1)~Find the
balls of radius $2i-1$ in $G$ centered at the endpoints of $e$. (2)~Check if $e'\in
M_{i-1}$ for each edge $e'$ incident to vertices in the balls. We can then
exhaustively check for each path $p$ of length $2i-1$ that contains $e$ whether $p$
is an $M_{i-1}$-augmenting path.

The \mis-step answers a query $p\in P^*_i$ by simulating
the \mis\ \clocal-algorithm over $H_i$. The \mis-algorithm needs to simulate
probes to $H_i$. A probe to $H_i$ consists of an
$M_{i-1}$-augmenting path $q$ and a port number. We suggest
to implement this probe by probing all the neighbors of $q$
in $H_i$ (so the port number does not influence the first
part of implementing a probe). See Algorithm~\ref{alg:prob
Hi}. As in the listing step, a probe $q$ in $H_i$ can be
obtained by (1)~finding the balls in $G$ of radius $2i-1$
centered at endpoints of edges in $E(q)$, and (2)~finding
out which edges within these balls are in $M_{i-1}$. The
first two steps enable us to list all of the neighbors of
$q$ in $H_i$ (i.e., the $M_{i-1}$-augmenting paths that
intersect $q$). These neighbors are ordered (e.g., by
lexicographic order of the node IDs along the path). If the
probe asks for the neighbor of $q$ in port $i$, then the
implementation of the probe returns the $i$th neighbor of
$q$ in the ordering.

By combining the recursive local improvement technique with
our deterministic stateless \clocal\ \mis-algorithm, we
obtain a deterministic stateless \clocal-algorithm that
computes a $(1-\eps)$-approximation for maximum matching.
The algorithm is invoked by calling the oracle
$\oracle_{k+1}$.
\begin{lemma}\label{lemma:mcm}
  The oracle $\oracle_{i}(e)$ is a $\clocal[2^{\Delta^{O(i)}}\cdot (\log^*n)^{i}]$ that
  computes whether $e\in M_{i}$.
\end{lemma}
\begin{proof}
  Correctness follows by induction on $i$ that shows that the oracle simulates
  Algorithm~\ref{alg:global mcm'}. We analyze the number of probes as follows.  To
  simplify notation, we denote the number of probes performed by algorithm $B$ by
  $|B|$, for example, $|\oracle_i|$ and $|\proca_i|$ denote the number of probes to
  $G$ performed by the oracle $\oracle_i$ and procedure $\proca_i$, respectively.
  Let $n_i$ and $\Delta_i$ denote the number of vertices and the maximum degree of
  $H_i$, respectively.

The probe complexity of $\oracle_i$ satisfies
  the following recurrence:
  \begin{align*}
    |\oracle_i|&=
    \begin{cases}
      0 & \text{if $i=0$},\\
      |\oracle_{i-1}| + |\proca_i| & \text{if $i\geq 1$}.\\
\end{cases}
\end{align*}

The probe complexity of $\proca_i$ is upper bounded as follows.  In Lines 2-3, each
BFS performs $O(\Delta^{2i})$ probes.  The number of edges in the probed ball is
$O(\Delta^{2i+1})$, and for each such edge a call to $\oracle_{i-1}$ is made in Line
4. Line 5 does not generate any probes.  Let $|\lmis_G(H_i)|$ denote the probe
complexity of the simulation of \clocal-algorithm for \mis\ over the intersection
graph $H_i$ when the access is to $G$. In Line 7, the number of probes is
bounded by $|P_i(e)|\cdot |\lmis_G(H_i)|$.  Hence,
\begin{align*}
 |\proca_i|
  &\leq    O(\Delta^{2i+1}) \cdot |\oracle_{i-1}| +
  |P_i(e)|\cdot |\lmis_G(H_i)|.
\end{align*}
The number of paths in $P_i(e)$ is at most $2i \cdot \Delta^{2i}$ (indeed, there are
$2i$ possibilities for the position of $e$ along a path, and, for each position $j$,
there are a most $\Delta^{j}\cdot\Delta^{2i-j}$ paths $p$ such that $e$ is the $j$th
edge in $p$).

We bound $|\lmis_G(H_i)|$ by the probe complexity $|\lmis_{H_i} (H_i)|$ (namely, the
probe complexity if one can access $H_i$) times the probe complexity of
simulating probes to $H_i$ via probes to $G$. By Corollary~\ref{coro:mis},
$|\lmis_{H_i} (H_i)|\leq \Delta_i^{O(\Delta_i^2)} \cdot \log^* n_i$.  Simulation of
probes in $H_i$ via probes to $G$ is implemented by the $\probe(i,p)$ procedure.
Similarly, to the analysis of the probe complexity of $\proca_i$, the probe
complexity of $\probe(i,p)$ is $O(2i \cdot \Delta^{2i+1}\cdot |\oracle_{i-1}|)$.

Hence,
\begin{align*}
    |P_i(e)|\cdot |\lmis_G(H_i)| & \leq 2i \cdot \Delta^{2i} \cdot
    \Delta_i^{O(\Delta_i^2)} \cdot \log^* n_i \cdot 2i \cdot \Delta^{2i+1}\cdot |\oracle_{i-1}|.
\end{align*}
Because $n_i\leq n^{2i}$ and $\Delta_i  = O(i^2 \cdot \Delta^{2i})$,
it follows that
\begin{align*}
    |\proca_i| & \leq
    \Delta^{\Delta^{O(i)}} \cdot \log^* n \cdot |\oracle_{i-1}|.
\end{align*}

We conclude that $|\oracle_i|$ satisfies
\begin{align*}
  |\oracle_i| &\leq  \Delta^{\Delta^{O(i)}} \cdot \log^* n \cdot |\oracle_{i-1}| \\
& \leq
\Delta^{\Delta^{O(i)}} \cdot (\log^* n)^i.
\end{align*}
Note that $\Delta^{\Delta^{O(i)}}=2^{\Delta^{O(i)}}$, and the lemma follows.
\end{proof}

\noindent\medskip
By setting $i = \lceil \frac 1\eps\rceil +1$, we obtain the following theorem.
\begin{theorem}\label{thm:mm}
  There is a deterministic, stateless, $(1-\eps)$-approximate
  \clocal$[\varphi]$-algorithm for maximum matching, where
  \[
    \varphi=(\log^* n)^{\lceil \frac 1\eps\rceil +1}\cdot 2^{\Delta^{O(1/\eps)}}.
  \]
\end{theorem}

%
\begin{algorithm}[H]
\caption{$\oracle_{i}(e)$ - a recursive oracle for
membership in the approximate matching.} \label{alg:Oi}
\begin{algorithmic}[1]
\Require A query $e \in E$.
\Ensure Is $e$ an edge in the matching $M_i$?
  \State If $i=0$ then return false.
  \State $\tau\gets \oracle_{i-1}(e)$.
  \State $\rho\gets \proca_{i}(e)$.
  \State \textbf{Return} $\tau \oplus \rho$.
\end{algorithmic}
\end{algorithm}
%
\begin{algorithm}[H]
\caption{$\proca_i(e=(u,v))$ - a procedure for checking
membership of an edge $e$ in one
  of the paths in $P^*_i$.} \label{alg:Ai}
\begin{algorithmic}[1]
\Require An edge $e \in E$.
\Ensure Does $e$ belong to a path $p\in P^*_i$?
\State \textbf{Listing step:} \Comment Compute all shortest $M_{i-1}$-augmenting
paths that contain $e$.
  \State \quad $B_u \gets BFS_G(u)$ with depth $2i-1$.
  \State \quad $B_v \gets BFS_G(v)$ with depth $2i-1$.
  \State \quad For every edge $e'$ in the subgraph of $G$ induced by $B_u\cup B_v$: $\chi_{e'}\gets \oracle_{i-1}(e')$.
  \State \quad $P_i(e)\gets$ all $M_{i-1}$-augmenting paths of length $2i-1$ that
  contain $e$ (based on information gathered in Lines 2-4).
\State \textbf{\mis-step:} \Comment Check if one of the augmenting paths is in $P^*_i$.
  \State \quad For every $p\in P_i(e)$: If $\lmis(H_i,p)$ \textbf{Return} true.
  \State \quad \textbf{Return} false.
\end{algorithmic}
\end{algorithm}
%
\begin{algorithm}[H]
  \caption{$\textit{probe}(i,p)$ - simulation of a probe to the intersection graph
    $H_i$ via probes to $G$. The probe returns all the $M_{i-1}$-augmenting paths
    that intersect $p$. } \label{alg:prob Hi}
\begin{algorithmic}[1]
\Require A path $p \in P_i$ and the ability to probe $G$.
\Ensure The set of $M_{i-1}$-augmenting paths of length $2i-1$ that intersect $p$.
\State For every $v\in p$ do
  \State \quad $B_v \gets BFS_G(v)$ with depth $2i-1$.
  \State \quad For every edge $e'\in B_v$: $\chi_e\gets \oracle_{i-1}(e)$. \Comment determine whether the path is alternating and whether the endpoints are $M_{i-1}$-free.
  \State \quad $P_i(v)\gets$ all $M_{i-1}$-augmenting paths of length $2i-1$ that contain $v$.
\item \textbf{Return} $\bigcup_{v\in p} P_i(v)$. 
\end{algorithmic}
\end{algorithm}

\section{A  \dlocal\ Approximate \mcm\ Algorithm}\label{sec:dlocal mcm}
In this section, we present a \dlocal-algorithm that
computes a $(1-\eps)$-approximate maximum cardinality
matching.  The algorithm is based on bounding the probe
radius of the $\clocal$-algorithm from Theorem~\ref{thm:mm}
and applying the simulation from
Proposition~\ref{prop:simul}.
\begin{theorem}\label{thm:distalg}
  There is a deterministic $\dlocal[\Delta^{O(1/\eps)} +
  O\left(\frac{1}{\eps^2}\right) \cdot\log^*(n)]$-algorithm for computing a
  $(1-\eps)$-approximate \mum.
\end{theorem}
\begin{proof}
  The proof of the theorem is based on the simulation of a \clocal-algorithm by a
  \dlocal-algorithm from Proposition~\ref{prop:simul}.  In Lemma~\ref{lemma:radius} we prove
  that the probe radius of $\oracle_k$ is $\Delta^{O(k)} + O(k^2)
  \cdot\log^*(n)$. Plug $k=1+\lceil \frac{1}{\epsilon}\rceil$,  and the theorem follows.
\end{proof}
\begin{lemma}\label{lemma:radius}
The probe radius of the \clocal-algorithm $\oracle_{k}$ is
  \[
    r_G(\oracle_k) = \Delta^{O(k)}+O(k^2)\cdot \log^*(n)\:.
  \]
\end{lemma}
\begin{proof}
The probe radius $r_G(\oracle_i)$ satisfies
  the following recurrence:
  \begin{align*}
    r_G(\oracle_i)&=
    \begin{cases}
      0 & \text{if $i=0$},\\
      \max\{r_G(\oracle_{i-1}), r_G(\proca_i)\}& \text{if $i\geq 1$.}
\end{cases}
\end{align*}

The description of the procedure $\proca_i$ implies that the probe radius $r_G(\proca_i)$ satisfies
the following recurrence:
  \begin{align*}
    r_{G}(\proca_i)&\leq
    \max\{2i+r_G(\oracle_{i-1}), 2i-1+r_G(\lmis(H_i))\}
\end{align*}

We bound the probe radius $r_G(\lmis(H_i))$ by composing the
radius $r_{H_i}(\lmis(H_i))$ with the increase in radius
incurred by the simulation of probes to $H_i$ by probes to
$G$.  Recall that the $\lmis$-algorithm is based on a
deterministic coloring algorithm $C$.  We denote the number
of colors used by $C$ to color a graph $G'$ by $|C(G')|$.

The \mis-algorithm orients the edges by coloring the vertices. The
radius of the orientation is at most the number of colors. It follows that
\begin{align*}
  r_{H_i}(\lmis(H_i)) &\leq r_{H_i}(C(H_i)) + |C(H_i)|.
\end{align*}
The simulation of probes to $H_i$ requires an increase in the probe radius.  In
general, suppose that algorithm $L$ probes $H$, and algorithm $S$ simulates probes to
$H$ by probes to $G$. Let $S(p)$ denote the set of probes in $G$ performed by $S$ to
simulate a probe of $p$ in $H$.  Suppose that $S(p)\cap S(p')\neq \emptyset$ whenever
$p$ and $p'$ are neighbors in $H$. In this case the probe radius of the composed
algorithm is at most $r_H(L)\cdot r_G(S)$.  However, our case is special in the
following sense. Consider a path $p_1,p_2,\ldots,p_r$ of length $r$ in $H_i$.  This
sequence $\{p_j\}$ of probes in $H$ is simulated by probes in $G$ by the procedure
$\probe(i,p_j)$, for $j=1,\ldots,r$.  The probe radius in $G$ from any vertex in
$p_1$ is bounded by $(2i)\cdot r + r_G(\probe(i-1))$.  Hence,
\begin{align*}
  r_{G}(\lmis(H_i)) &\leq 2i \cdot   r_{H_i}(\lmis(H_i)) + r_G(\probe (i-1)).
\end{align*}

Many distributed coloring algorithms find a vertex coloring in $O(\log^*(n) +
\poly(\Delta))$ rounds (giving us the same upper bound on the probe-radius of the
corresponding \clocal-algorithm) and use $\poly(\Delta)$ colors~(see,
for example,
\cite{barenboim2009distributed,linial,cole1986deterministic,panconesi2001some,kuhn2009weak}).
Plugging these parameters in the recurrences yields
\begin{align*}
  r_G(\oracle_i) &\leq 2i + r_G(\lmis(H_i))\\
&\leq 2i\cdot (1+  r_{H_i}(\lmis(H_i))) + r_G(\probe(i-1))\\
&\leq r_G(\oracle_{i-1}) + O\Big(i \cdot \log^* (n_i) + \poly(\Delta_i)\Big),
\end{align*}
Since $\Delta_i \leq (2i)^2 \Delta^{2i-1}$ and $n_i \leq n^{2i}$ we get that
  \begin{align*}
    r_G(\oracle_k) & \leq \sum_{i=1}^{k} O\left(i\cdot \log^*(n) + \poly((2i)^2\cdot \Delta^{2i})\right) \\
    & = O(k^2\cdot \log^*(n)) + \Delta^{O(k)}.
  \end{align*}
The lemma follows.\end{proof}

\section{A Global $(1-\eps)$-Approximate \mwm\ Algorithm}\label{sec:global mwm}
In this section we present a deterministic stateless \clocal-algorithm that computes
a $(1-\eps)$-approximation of a maximum weighted matching.\footnote{To avoid dealing
  with constants, we present a $1-O(\eps)$-approximation.}  The algorithm is based on
a parallel $(1-\eps)$-approximation algorithm for weighted matching of Hougardy and
Vinkemeier~\cite{HV06}.

\paragraph{Terminology and Notation.}
In addition to the terminology and notation used in the unweighted case, we define
the following terms. In the weighted case, a path is $M$-alternating if it is a
simple path or a simple cycle in which the edges alternate between $M$ and
$E\setminus M$.  For a matching $M$ and an $M$-alternating path $p$, the
\emph{gain} of $p$ is defined by
\[
    \gain_M(p) \eqdf w(p\setminus M) - w(p \cap M)\:.
\]
The gain of a set of (disjoint) paths is the sum of the gains of the paths in the set.

An $M$-alternating path $p$ is \emph{$M$-augmenting} if $\gain_M(p)>0$ and $p$ satisfies
one of the following conditions: (1)~$p$ is a simple cycle, or (2)~$p$ is a simple
path that satisfies: if $p$ ends (or begins) in an edge not in $M$, then the
corresponding endpoint is $M$-free. Note that the symmetric difference between $M$ and any
set of vertex disjoint $M$-augmenting paths is a matching with higher weight.

 We say that a path $p$ is \emph{$(M,[1,k])$-augmenting} if $p$ is $M$-augmenting and
 $|E(p) \setminus M| \leq k$. An $(M,[1,k])$-augmenting path may contain at most
 $2k+1$ edges ($k$ non-matching edges and $k+1$ matching edges).  The \emph{gain-index} of an $M$-augmenting path $p$ is defined by
\[
\gamma_M(p)\triangleq \left\lceil\log_2 \gain_M(p) \right\rceil.
\]

Let $I(M)$ denote the intersection graph of $(M,[1,k])$-augmenting paths.  Namely,
the vertices of $I(M)$ are the $(M,[1,k])$-augmenting paths, and two vertices in
$I(M)$ are neighbors if they have a common vertex in $G$.  We partition the vertices of
$I(M)$ (i.e., $[M,[1,k])$-augmenting paths of $G$) to classes; the \emph{class} of an
augmenting path equals its gain-index.

\paragraph{Optimal Set of Augmentation Paths.}
Given a matching $M$, let $\aug(M,k)$ denote a set of vertex disjoint
$(M,[1,k])$-augmentation paths with maximum gain. Equivalently, $\aug(M,k)$ is an \mis\
in $I(M)$ with maximum gain.

\begin{theorem}[\cite{DBLP:journals/ipl/PettieS04}]\label{thm:pettie}
Let $M$ and $M^*$ denote a matching and maximum weight matching in $G$, respectively,
then
\begin{align*}
  \gain(\aug(M,k)) \geq \frac{k+1}{2k+1} \cdot \left( \frac{k}{k+1} \cdot w(M^*) - w(M)\right).
\end{align*}\end{theorem}
\paragraph{Index-Greedy Augmentation.}
An \emph{index-greedy} set of augmentation paths is an
\mis\ in $I(M)$ obtained by the sequential \mis\ algorithm
where the vertices in $I(M)$ are sorted in non-increasing
gain-index order. We denote an index-greedy augmentation by
$\aug^{ig}(M,k)$.

The following proposition states that the gain of every index-greedy augmentation is
a $2(k+1)$-approximation of the gain of $\aug(M,k)$. It follows from the fact that a
greedy \mis\ is a $(k+1)$-approximation of a max-weight \mis\ if each vertex in the
greedy \mis\ intersects at most $(k+1)$ vertices from a max-weight \mis, and from the
fact that the ratio between gains of paths with the same gain-index is at most $2$.
\begin{proposition}\label{prop:gainindex}
  \begin{align*}
    \gain(\aug^{ig}(M,k)) \geq \frac {1}{2(k+1)} \cdot \gain(\aug(M,k)).
  \end{align*}
  \end{proposition}
\begin{proof}
  Let $\aug^{ig}(M,k)=\{p_1,\ldots,p_r\}$, where $\gamma_M(p_i)\geq \gamma_M(p_{i+1})$.
  Namely, $p_i$ is added to the index-greedy augmentation before $p_{i+1}$.  We partition
  $\aug(M,k)$ into disjoint sets $X_1\cup \cdots \cup X_r$ as follows.  Each
  augmentation path $q\in\aug(M,k)$ is in the set $X_i$ with the smallest index $i$
  such that $q=p_i$ or $q$ is a neighbor of $p_i$ (in the intersection graph $I(M)$).

  Since $X_1\cup\cdots \cup X_r$ is a partition of $\aug(M,k)$, it suffices to prove
  that $\min_i \frac{\gain(p_i)}{\gain(X_i)} \geq \frac{1}{2(k+1)}$. Indeed, this inequality follows from two
  facts. First, every $(M,[1,k])$-augmenting path intersects at most $k+1$ paths in
  $\aug(M,k)$. Second, by the ordering of the augmentations in non-increasing gain-index order, it
  follows that $\gain(p_i)\geq \frac 12 \cdot\gain(q)$, for every $q\in X(p_i)$.
\end{proof}

\paragraph{Outline of the Global Algorithm.}
The main differences between the global approximation algorithms for weighted and
unweighted matchings are: (1)~The length of the augmenting paths (and cycles) does
not grow; instead, during every step, $(M,[1,k])$-augmenting paths are used.  (2)~The
set of disjoint augmenting paths in each iteration in the weighted case is chosen
greedily, giving precedence to augmentations with higher gain-index.  We denote the
computation of an index-greedy augmentation by \igmis. The global algorithm is listed
as Algorithm~\ref{alg:global mwm}.

\paragraph{Algorithm Notation.} The global algorithm uses the following notation.
The algorithm computes a sequence of matchings $M_{i}$ (where $i\in[1,L]$, for
$L=O(\frac{1}{\eps} \log \frac{1}{\eps})$).  We denote the initial empty matching by
$M_{0}$.  Let $I(M_i)$ denote the intersection graph over $(M_i,[1,k])$-augmenting
paths with edges between paths whenever the paths share a vertex. The class of each
vertex in $I(M_i)$ (i.e., augmenting path in $G$) is the gain-index of the path. Let
$\igmis(I(M_i))$ denote a index-greedy \mis\ in $I(M_i)$ with precedence given to
vertices with higher gain-indexes.

\begin{algorithm}[H]
  \caption{$\text{Global-APX-MWM}(G,\eps)$ - a global version of the
    $(1-\eps)$-approximate \mwm\ Algorithm Hougardy and
    Vinkemeier~\cite{HV06}.} \label{alg:global mwm}
\begin{algorithmic}[1]
  \Require A graph $G=(V,E)$ with edge weights.

  \Ensure A $(1-O(\eps))$-approximate weighted matching

\State $k \gets\lfloor \frac{2}{\eps}\rfloor$.

\State $L \gets \lceil 2\cdot(2k+1)\cdot \ln (2/\eps)\rceil$.

\State $M_0\gets \emptyset$.

\For {$i=1$ to $L$}
\State Let $I(M_{i-1})$ denote the intersection graph of $(M_{i-1},[1,k])$-augmenting paths.
\State $\aug_i \gets \igmis(I(M_{i-1}))$, where the class of each augmenting path is its gain-index.
\State $M_{i}\gets M_{i-1} \oplus E(\aug_i)$.
\EndFor

\State \textbf{Return} $M_L$.
\end{algorithmic}
\end{algorithm}

\begin{algorithm}[H]
\caption{$\oracle_{i}(e)$ - a recursive oracle for
membership in the approximate weighted matching.} \label{alg:mwm oracle}
\begin{algorithmic}[1]
\Require A query $e \in E$.
\Ensure Is $e$ an edge in the matching $M_i$?
  \State If $i=0$ then return false.
  \State $\tau\gets \oracle_{i-1}(e)$.
  \State $\rho\gets \proca_{i}(e)$.
  \State \textbf{Return} $\tau \oplus \rho$.
\end{algorithmic}
\end{algorithm}

\begin{algorithm}[H]
\caption{$\proca_{i}(e=(u,v))$ - a procedure for checking membership of an edge $e$ in one
  of the paths in $\aug_i$.} \label{alg:mwm aug}
\begin{algorithmic}[1]
\Require An edge $e \in E$.
\Ensure Does $e$ belong to a path $p\in \aug_i$?
\State \textbf{Listing step:} \Comment Compute all shortest $M_{i-1}$-augmenting
paths that contain $e$.
  \State \quad $B_u \gets BFS_G(u)$ with depth $(2k+1)$.
  \State \quad $B_v \gets BFS_G(v)$ with depth $(2k+1)$.
  \State \quad For every edge $e'$ in the subgraph of $G$ induced by $B_u
\cup B_v$: $\chi_{e'}\gets \oracle_{i-1}(e')$.
  \State \quad $P_{i}(e)\gets$ all $(M_{i-1},[1,k])$-augmenting paths that
  contain $e$.

\State \textbf{\mis-step:}
\Comment Check if $e$ is in one of the augmenting paths is in $P^*_{i,j}$.
  \State \quad For every $p\in P_{i}(e)$: If $p\in \igmis(I(M_{i-1}))$ \textbf{Return} true.
  \State \quad \textbf{Return} false.
\end{algorithmic}
\end{algorithm}

\begin{algorithm}[H]
  \caption{$\textit{probe}(i-1,p)$ - simulation of a probe to the intersection graph
    $I(M_{i-1})$ via probes to $G$. } \label{alg:mwm probe}
\begin{algorithmic}[1]
\Require An $(M_{i-1},[1,k])$-augmenting path $p \in I(M_{i-1})$ and the ability to probe $G$.
\Ensure The set of $(M_{i-1},[1,k])$-augmenting paths  that intersect $p$ (i.e., neighbors of
$p$ in $I(M_{i-1})$).
\State For every $v\in p$ do
  \State \quad $B_v \gets BFS_G(v)$ with depth $2k+1$.
  \State \quad For every edge $e'\in B_v$: $\chi_e\gets \oracle_{i-1}(e)$. \Comment
  needed to determine whether a path is an $(M_{i-1},[1,k])$-augmenting path.
  \State \quad $P_i(v)\gets$ all $(M_{i-1},[1,k])$-augmenting paths that contain $v$.
\item \textbf{Return} $\bigcup_{v\in p} P_i(v)$. 
\end{algorithmic}
\end{algorithm}

\paragraph{Correctness.}
\begin{theorem}[\cite{HV06}]\label{thm:HV06}
  Algorithm~\ref{alg:global mwm} computes a $(1-\eps)$-approximate maximum weighted matching.
\end{theorem}
\begin{proof}
  By Propositions~\ref{prop:gainindex} the augmentations computed in Line 6 of the
  algorithm satisfy
  \begin{align}
    \label{aqarray:1}
    \gain(\aug_i) &\geq \frac{1}{2(k+1)} \cdot \gain(\aug(M_{i-1},k)).
  \end{align}
By Theorem~\ref{thm:pettie}
\begin{align*}
  \gain(\aug(M_{i-1},k)) &\geq \frac{k+1}{2k+1} \left(\frac{k}{k+1}\cdot w(M^*) - w(M_{i-1})\right).
\end{align*}
Let $\rho_i\triangleq w(M_i)/w(M^*)$. It follows that $\rho_i$ satisfies the
recurrence
\begin{align*}
  \rho_i &\geq \left(1-\frac{1}{2(2k+1)}\right)\rho_{i-1} + \frac{k}{k+1}\cdot \frac{1}{2(2k+1)}.
\end{align*}
Hence,
\begin{align*}
  \rho_L & \geq \frac{k}{k+1}\cdot \frac{1}{2(2k+1)}\cdot
  \frac{1-\left(1-\frac{1}{2(2k+1)}\right)^L}
         {1-(1-\frac{1}{2(2k+1)})}\\
&=\frac{k}{k+1}\cdot \left( 1- \left(1-\frac{1}{2(2k+1)} \right)^L \right).
\end{align*}
The theorem follows by setting $k=\Theta\left(\frac{1}{\eps}\right)$ and
$L=\Theta(\frac{1}{\eps}\log \frac{1}{\eps})$.
\end{proof}
\section{A \clocal\ $(1-\eps)$-Approximate \mwm\ Algorithm}\label{sec:clocal mwm}
In this section we present a $\clocal$-algorithm that implements the global
$(1-\eps)$-approximation algorithm for \mwm.

\subsection{Preprocessing}
We assume that the maximum edge weight is known (as well as $n$,$\eps$, and
$\Delta$).  By normalizing the weights, we obtain that the edge weights are in the
interval $(0,1]$. Note, that at least one edge has weight $1$.

We round down the edge weights to the nearest integer multiple of $\eps/n$. Let
$w(e)$ denote the original edge weights and let $w'(e)$ denote the rounded down
weights. Therefore, $w(e)- \eps/n< w'(e)\leq w(e)$. Note that as a result of rounding
down edge weights, the minimum positive weight is at least $\eps/n$.  For every
matching $M$, we have $w(M) - \eps/2 \leq w'(M)$.  As there exists one edge of weight
$1$, the effect of discretization of edge weights decreases the approximation factor
by at most a factor of $(1-\eps/2)$.

\paragraph{Number of Distinct Gain-Indexes.} The rounded edge weights are multiples
of $\eps/n$ in the interval $[\eps/n,1]$.  Let
$$\wmin(\eps)\triangleq\min\{w(e) \mid  w(e)\geq\eps/n\}.$$
Note that $\wmin(\eps)\geq \eps/n$.\footnote{We remark
that $\wmin(\eps)$ may be much
  bigger than $\eps/n$. For example, if $\wmin$ is constant (say, $1/100$). The
  analysis of the probe complexity and the probe radius uses $1/\wmin(\eps)$ instead
  of $n/\eps$ to emphasize the improved results whenever $1/\wmin(\eps)$ is
  significantly smaller than $2n/\eps$.}

As the edge weights are multiples of $\eps/n$ in the
interval $[\wmin(\eps),1]$, it follows that the gains of
$(M,[1,k])$-augmenting paths are in the range
$[\wmin(\eps),k]$.  Hence $(M,[1,k])$-augmenting paths have
at most $O(\log (k/\wmin(\eps))$ distinct gain-indexes.

\subsection{\clocal-Implementation}
\paragraph{\clocal-algorithm for index-greedy \mis.}
A sequential algorithm for computing an index-greedy \mis\ of $G$ adds vertices to
the \mis\ by scanning the vertices in nonincreasing gain-index order.  We refer to
this algorithm as $\igmis$.  Following Section~\ref{sec:lin}, a simulation of such a
sequential algorithm is obtained by computing an acyclic orientation. For
\igmis, the orientation is induced by the vertex coloring that is the
Cartesian product of the gain-index of the vertex and its (regular) color.  Lexicographic
ordering is used to compare the colors.  We summarize the probe complexity and probe
radius of the \clocal-algorithm for \igmis\ in the following lemma (recall that $\ell$ denotes
the number of distinct index-gains).
\begin{lemma}\label{lemma:lexmis}
  An index-greedy \mis\ can be computed by a $\clocal$-algorithm with the following
  properties:
  \begin{enumerate}
  \item The probe radius is $O(\Delta^2\cdot \ell+\log^*n)$.
  \item The probe complexity is $O(\Delta^{\Delta^2\cdot \ell+1}\cdot (\log^* n +
    \Delta^2))=\Delta^{O(\Delta^2\cdot \ell)}\cdot \log^*n$.
  \end{enumerate}
\end{lemma}
\begin{proof}
  The probe radius is simply the number of colors (in the Cartesian product) plus the
  radius of the regular $\Delta^2$-coloring algorithm. The number of colors is $\Delta^2 \cdot
  \ell$ and the radius of the $\Delta^2$-coloring algorithm is $O(\log^*n)$.

  The probe complexity is bounded by the reachability of the orientation times the
  probe complexity of the regular $\Delta^2$-coloring algorithm. The reachability of
  the orientation is bounded by $\Delta^{\Delta^2\cdot \ell}$. The probe complexity
  of the regular $\Delta^2$-coloring algorithm is $\Delta \cdot (\log^* n +
  \Delta^2)$, and the lemma follows.
\end{proof}
Note that the \clocal-algorithm computes a \igmis\ over $I(M_i)$ in which the class
of a vertex equals its gain-index. As there are $O(\log(k/\wmin(\eps)))$ distinct
gain-indexes, it follows that we can apply Lemma~\ref{lemma:lexmis} with
$\ell=O(\log(k/\wmin(\eps)))$ and $\Delta=\Delta(I(M_i))$.

\paragraph{\clocal\ implementation of the global algorithm.}
The implementation also uses the local improvement technique of Nguyen and
Onak~\cite{onak2008} repeating the same method used in Section~\ref{sec:clocal mcm}.
The pseudo-code of the \clocal-algorithm appears as Algorithms~\ref{alg:mwm
  oracle}-\ref{alg:mwm probe}. This \clocal-algorithm implements
Algorithm~\ref{alg:global mwm}.

\medskip\noindent By induction, one can
prove that $\oracle_i(e)$ computes membership of $e$ in $M_i$. From Theorem~\ref{thm:HV06} we
obtain that $\oracle_L$ is an $(1-\eps)$-approximate $\clocal$-algorithm. The
following theorem analyzes the probe complexity of $\oracle_L$ (the theorem holds
under the assumption that $\wmin(\eps)<1$).
\begin{theorem}\label{thm:clocal mwm}
  There exists a $\clocal[\varphi]$-algorithm for $(1-\eps)$-approximate maximum
  weighted matching with
  \begin{align*}
    \varphi &= \left(\frac{1}{\wmin(\eps)}\right)^{\Delta^{O(1/\eps)}} \cdot
    (\log^*n)^{O(\frac{1}{\eps}\cdot \log\frac{1}{\eps})}
  \end{align*}
\end{theorem}
\begin{proof}
  The analysis is similar to the one in Lemma~\ref{lemma:mcm}. The key differences
  are as follows: (1)~The augmenting paths in all recursive calls have length at most
  $k$. Hence the intersection graph is not the same graph in both algorithms. (2)~A
  lexicographic-\mis\ is computed instead of an \mis. The analysis
  proceeds as follows.
\begin{align*}
  |\oracle_i| &\leq  |\oracle_{i-1}| + |\proca_i| \\
  &\leq  |\oracle_{i-1}| +  2\cdot \Delta^{2k+1} + \Delta^{2k+2} \cdot |\oracle_{i-1}| +
  |P_i(e)|\cdot |\igmis(I(M_{i-1}))|\cdot |\probe(i)|\\
&\leq \Delta^{O(k)}\cdot |\oracle_{i-1}| + \Delta^{O(k)} \cdot
(\Delta_i^{O(\Delta_i^2\cdot\log (k/\wmin(\eps)))} \cdot \log^* n_i) \cdot ( \Delta^{O(k)}\cdot |\oracle_{i-1}|).
\end{align*}
Because $n_i\leq n^{2k+1}$ and $\Delta_i  = O((2k+1)^2 \cdot \Delta^{2k+1})$,
it follows that
\begin{align}
\nonumber
  |\oracle_L| &\leq  \Delta^{\Delta^{O(k)}\cdot \log(1/\wmin(\eps))} \cdot \log^* n \cdot |\oracle_{L-1}|
\\
  \label{eq:Qi}
&=\Delta^{\Delta^{O(k)}\cdot \log(1/\wmin(\eps))} \cdot (\log^* n)^L.
\end{align}
Note that $\Delta^{\log(1/\wmin(\eps))}=(1/\wmin(\eps))^{\log(\Delta)}$, and the lemma follows.

\end{proof}

\section{A \dlocal\ $(1-\eps)$-Approximate \mwm\ Algorithm}
In this section, we present a \dlocal-algorithm that
computes a $(1-\eps)$-approximate weighted matching. The
algorithm is based on the same design methodology as in
Section~\ref{sec:dlocal mcm}. Namely, we bound the probe
radius of the \clocal-algorithm for \mwm\ (see
Lemma~\ref{lemma:radius weight}) and apply the simulation
technique (see Proposition~\ref{prop:simul}).

\begin{theorem}\label{thm:distalg mwm}
  There is a deterministic
$\dlocal[r ]$
-algorithm for computing a
  $(1-\eps)$-approximate \mwm\ with
  \[
  r_G(\oracle_L) \leq
O\left(\frac{1}{\eps^2} \cdot \log \frac{1}{\eps}\cdot \log^*n\right) +
    \Delta^{O(1/\eps)}\cdot \log \left(\frac{1}{\wmin(\eps)}
         \right)
  \]
\end{theorem}
\medskip\noindent
The proof of Theorem~\ref{thm:distalg mwm} is based on the following lemma.
Recall that ignoring lightweight edges implies that $\frac{1}{\wmin(\eps)} \leq
\frac{n}{\eps}$.

\begin{lemma}\label{lemma:radius weight}
The probe radius of the \clocal-algorithm $\oracle_{L}$ is
  \[
  r_G(\oracle_L) \leq
O\left(\frac{1}{\eps^2} \cdot \log \frac{1}{\eps}\cdot \log^*n\right) +
    \Delta^{O(1/\eps)}\cdot \log \left(\frac{1}{\wmin(\eps)}
         \right)
  \]
\end{lemma}

\begin{proof}
  The description of the oracle $\oracle_i$ implies that the probe radius
  $r_G(\oracle_i)$ satisfies the following recurrence:
  \begin{align*}
    r_G(\oracle_i)&=
    \begin{cases}
      0 & \text{if $i=0$},\\
      \max\{r_G(\oracle_{i-1}), r_G(\proca_i)\}& \text{else}.
\end{cases}
\end{align*}
The description of the procedure $\proca_i$ implies
that the probe radius $r_G(\proca_i)$ satisfies the
following recurrence:
  \begin{align*}
    r_{G}(\proca_i)&\leq O(k)+
\max\{r_G(\oracle_{i-1}), r_G(\igmis(I(M_{i-1})))\}.
\end{align*}
The probe radius of \igmis\ with respect to $G$ satisfies
\begin{align*}
  r_{G}(\igmis(I(M_{i-1}))) &\leq
  O(k)\cdot r_{I(M_{i-1})}(\igmis(I(M_{i-1}))) + r_G(\textit{probe}(i-1,p)).
\end{align*}
By Lemma~\ref{lemma:lexmis}, $r_{I(M_{i-1})}(\igmis(I(M_{i-1})))\leq O(\log ^*n) +
\Delta^{O(k)}\cdot \log (1/\wmin(\eps))$.

The probe radius of a simulation of a probe to $I(M_{i-1})$ satisfies
\begin{align*}
  r_G(\textit{probe}(i-1,p))& \leq O(k)+r_G(\oracle_{i-1}).
\end{align*}

It follows that
\begin{align*}
  r_G(\oracle_i) &\leq r_G(\oracle_{i-1})+
O(k\cdot \log^*n) + \Delta^{O(k)}\cdot \log (1/\wmin(\eps))\\
&\leq i\cdot \left(O(k\cdot \log^*n) + \Delta^{O(k)}\cdot \log (1/\wmin(\eps))\right),
\end{align*}
and the lemma follows.
\end{proof}

\section{Upper Bounds and Lower Bounds for \orad\ in the \dlocal\ Model}
In this section we consider $\dlocal$-algorithms for computing orientations over
bounded degree graphs.  The goal is to find an orientation with the smallest possible
radius (\orad). We first list \dlocal$[\log^*n]$-algorithms for $\orad$ that are
obtained from vertex coloring algorithms in which the radius of the orientation is
polynomial in the maximum degree of the graph. We then prove that every orientation
that computed in $o(\log^* n)$ rounds must have a radius that grows as a function of
$n$. Thus, $\Theta(\log^* n)$ rounds are necessary and sufficient for computing an
acyclic orientation with reachability that is bounded by a function of the maximum
degree.

\subsection{\dlocal\ Algorithms for \orad}\label{sec:distobr}
As observed in Proposition~\ref{prop:color2ori}, every
vertex coloring induces an acyclic orientation. This
implies that a \dlocal\ $c$-coloring algorithm can be used
to compute an acyclic orientation with radius $c$ by
performing the same number of rounds .  The distributed
coloring
algorithms~\cite[Theorem~4.2]{linial}~\cite[Theorem~4.6]{barenboim2009distributed}
imply the following corollary.
\begin{coro}\label{coro:dobr}
There are \dlocal\ algorithms for \orad\ with the following parameters:
\begin{enumerate}
  \item Radius $O({\Delta^2})$ in $O(\log^* n)$ rounds.
  \item Radius $\Delta+1$ in $O(\Delta) + \frac 12 \log^* n$ rounds.
\end{enumerate}
\end{coro}

\subsection{Lower Bound for \orad\ in the \dlocal\ Model}\label{sec:lb}
In this section we consider the problem of computing an acyclic orientation $H$ of a graph
$G$ with radius $\rad(H)$ that does not depend on the number of vertices
$n$ (it may depend on $\Delta$).
\begin{definition}
  Let $g:\NN \rightarrow \NN$ denote a function.  In the $\orad(g)$-problem, the input
  is a graph $G$ with maximum degree $\Delta$. The goal is to compute an orientation $H$ of
  $G$ with radius $\rad(H)\leq g(\Delta)$ (if such an orientation exists).
\end{definition}
\noindent \medskip
Our goal is to prove the following theorem.
\begin{theorem}\label{thm:lbdistobr}
For every function $g$, there is no \dlocal$[o(\log^* n)]$-algorithm that solves the $\orad(g)$-problem.
\end{theorem}
\begin{proof}
  The proof is based on a reduction from $\mis$ to $\orad$.
  Let $G_n$ denote an undirected ring with $n$ vertices.  Let $g:\NN \rightarrow
  \NN$ be any function (e.g., Ackermann function).  Assume, for the sake of
  contradiction, that there exists a $\dlocal[r]$-algorithm that computes an acyclic
  orientation $H_n$ of $G_n$ with radius $\rad(H_n) \leq g(\Delta)$.  Then,
  by Lemma~\ref{lemma:simul} and
    Proposition~\ref{prop:simul} there is a \dlocal$[r+g(\Delta)]$-algorithm for \mis.

  If $r=o(\log^*n)$, then this contradicts the theorem of Linial~\cite{linial} that
  states that there is no \dlocal\ algorithm that computes an \mis\ over a ring in less
  than $\frac 12 \cdot\log^* n$ rounds.
\end{proof}
\begin{rem}
Theorem~\ref{thm:lbdistobr} can be extended to $g:\NN
\times \NN \rightarrow \NN$ that is a function of $\Delta$
and $n$. The dependency on $n$ can be at most $o(\log^*
n)$, while the dependency on $\Delta$ stays arbitrary.
\end{rem}
\begin{rem}
  Theorem~\ref{thm:lbdistobr} can be extended to randomized algorithms since the lower
  bound for \mis\ in~\cite{linial} holds also for randomized algorithms.
\end{rem}

\section{Discussion}

In this work we design centralized local algorithms
for several graph problems.
Our algorithms are deterministic, do not use any state-space, and the number of probes
(queries to the graph) is $\poly(\log^* n)$ where $n$ is the number of graph
vertices.\footnote{For approximate weighted matching, we require a constant ratio of
  maximum-to-minimum edge weight.}   Previously known algorithms for these problems make ${\rm polylog}(n)$
probes, use ${\rm polylog}(n)$ state-space, and have failure probability $1/\poly(n)$.
While a basic tool in previous works is (random) {\em vertex rankings\/}, our basic
(seemingly weaker) tool, is {\em acyclic graph orientations with bounded
  reachability\/}. That is, our algorithms use as a subroutine a local procedure that
orients the edges of the graph while ensuring an upper bound on the number of
vertices reachable from any vertex. To obtain such orientations we employ a local
{\em coloring algorithm\/} which uses techniques from local {\em distributed\/}
algorithms for coloring.

On the other hand, by using a technique of Nguyen and Onak~\cite{onak2008} that was introduced
for local computation in the context of sublinear approximation algorithms, we get a
new result in local distributed computing: A deterministic algorithm for
approximating a maximum matching to within $(1-\eps)$ that performs
$\Delta^{O(1/\eps)}+ O\left(\frac{1}{\eps^2}\right)\cdot \log^* n$ rounds where
$\Delta$ is the maximum degree in the graph.  This is the best known algorithm for
this problem for constant $\Delta$. The technique also extends to approximate maximum
weighted matching.

The probe complexity of any \clocal-algorithm $A$ is bounded by $\Delta^{\rad(A)}$,
where $\rad(A)$ denote the probe radius of $A$.  Employing the above bound on the
probe complexity of our \clocal-algorithms places the $\log^* n$ in the exponent.
Our analyses of the probe complexity in the \clocal-algorithms is slightly stronger
because it avoids having the $\log^* n$ in the exponent.

\bibliographystyle{alpha}
\bibliography{local}

\end{document}